\documentclass[twocolumns]{IEEEtran}

\usepackage[colorlinks,urlcolor=blue,linkcolor=blue,citecolor=blue]{hyperref}

\usepackage{color,array}

\usepackage{graphicx}

\usepackage{cite}
\usepackage{amsmath,amssymb,amsfonts}
\usepackage{algorithmic}
\usepackage{graphicx}
\usepackage{textcomp}
\usepackage{bm}

\usepackage{hyperref}
\usepackage{siunitx}
\usepackage[capitalise]{cleveref}
\usepackage{soul}
\usepackage{amsthm}

\usepackage{empheq}
\DeclareMathAlphabet\mathbfcal{OMS}{cmsy}{b}{n}

\newcommand*{\genbf}[1]{\ifmmode\mathbf{#1}\else\textbf{#1}\fi}
\newcommand{\norm}[1]{\left\lVert#1\right\rVert}

\newcommand{\Phix}[2]{\bm{\Phi}^{\mathbf{x}}[\mathbf{#1},\mathbf{#2}]}
\newcommand{\Phiu}[2]{\bm{\Phi}^{\mathbf{u}}[\mathbf{#1},\mathbf{#2}]}
\newcommand{\Psix}[0]{\bm{\Psi}^{\mathbf{x}}}
\newcommand{\Psiu}[0]{\bm{\Psi}^{\mathbf{u}}}
\newcommand{\Phixno}[0]{\bm{\Upsilon}^{\mathbf{x}}}
\newcommand{\Phiuno}[0]{\bm{\Upsilon}^{\mathbf{u}}}

\newcommand{\Emme}[0]{\mathbfcal{M}}
\newcommand{\Effe}[0]{\mathbfcal{F}}
\newcommand{\LL}{{\mathcal{L}}}

\newcommand{\xb}[0]{\mathbf{x}}
\newcommand{\ub}[0]{\mathbf{u}}
\newcommand{\Fb}[0]{\mathbf{F}}
\newcommand{\wb}[0]{\mathbf{w}}

\newcommand{\hatbf}[1]{\widehat{\mathbf{#1}}}

\newcommand{\bmx}{\bold{x}}
\newcommand{\bmu}{\bold{u}}
\newcommand{\bmw}{\bold{w}}
\newcommand{\bmF}{\bold{F}}
\newcommand{\bmK}{\bold{K}}
\newcommand{\bmC}{\bold{C}}

\newcommand{\bmcalM}{\boldsymbol{\mathcal{M}}}

\newcommand{\cG}{\mathcal{G}}
\newcommand{\cV}{\mathcal{V}}
\newcommand{\cE}{\mathcal{E}}

\usepackage{tikz}
\usetikzlibrary{arrows}

\makeatletter
\newcommand\@erelb@r[1]{%
  \mathrel{\tikz[baseline=-.5ex]\draw[#1] (0,0)---(0.3,0);}
}
\newcommand{\erelbar}[1]{\@erelbar#1}
\def\@erelbar#1#2{%
  \ifcase\numexpr#1*4+#2\relax
    \@erelb@r{-}\or     
    \@erelb@r{->}\or    
    \@erelb@r{-|}\or    
    \@erelb@r{->|}\or   
    \@erelb@r{<-}\or    
    \@erelb@r{<->}\or   
    \@erelb@r{<-|}\or   
    \@erelb@r{<->}\or   
    \@erelb@r{|-}\or    
    \@erelb@r{|->}\or   
    \@erelb@r{|-|}\or   
    \@erelb@r{|<->|}\or 
    \@erelb@r{|<-}\or   
    \@erelb@r{|<->}\or  
    \@erelb@r{|<-|}\or  
    \@erelb@r{|<->|}    
  \else
    \@wrong
  \fi
}
\makeatother

\newtheoremstyle{mystyle}
  {}
  {}
  {}
  {}
  {\bfseries}
  {.}
  { }
  {}

\theoremstyle{mystyle}
\newtheorem{proposition}{Proposition}
\newtheorem{theorem}{Theorem}

\newtheorem{corollary}{Corollary}
\newtheorem{remark}{Remark}
\newtheorem{definition}{Definition}
\newtheorem{problem}{Problem}

\title{  \textbf{ Learning to Boost the Performance \\ of Stable Nonlinear Systems}} 
\begin{document}


\author{Luca Furieri, Clara Luc\'{i}a Galimberti, and Giancarlo Ferrari-Trecate
\thanks{ L. Furieri, C. L. Galimberti, and G. Ferrari-Trecate are with the Institute of Mechanical Engineering, EPFL, Switzerland. E-mail addresses: \{luca.furieri, clara.galimberti, giancarlo.ferraritrecate\}@epfl.ch.}
\thanks{Research supported by the Swiss National Science Foundation (SNSF) under the NCCR Automation (grant agreement 51NF40\textunderscore 80545). Luca Furieri is also grateful to the SNSF for the Ambizione grant PZ00P2\textunderscore208951.}
}
\maketitle

\begin{abstract}
The growing scale and complexity of safety-critical control systems underscore the need to evolve current control architectures aiming for the unparalleled performances achievable through state-of-the-art optimization and machine learning algorithms. However, maintaining closed-loop stability while boosting the performance of nonlinear control systems using data-driven and deep-learning approaches stands as an important unsolved challenge. In this paper, we tackle the performance-boosting problem with closed-loop stability guarantees. Specifically, we establish a synergy between the Internal Model Control (IMC) principle for nonlinear systems and state-of-the-art unconstrained optimization approaches for learning stable dynamics. Our methods enable learning over arbitrarily deep neural network classes of performance-boosting controllers for stable nonlinear systems; crucially, we guarantee $\mathcal{L}_p$ closed-loop stability even if optimization is halted prematurely, and even when the ground-truth dynamics are unknown, with vanishing conservatism in the class of stabilizing policies as the model uncertainty is reduced to zero. 
We discuss the implementation details of the proposed control schemes, including distributed ones, along with the corresponding optimization procedures, demonstrating the potential of freely shaping the cost functions through several numerical experiments.
\end{abstract}


\maketitle

\section{Introduction}
\label{sec:introduction}

The success of control systems across a broad spectrum of applications --- from manufacturing to water, power, and transportation networks \cite{annaswamy2023control} --- is rooted not only in advancements in sensing, computation, and communication but also in the growing availability of methods for designing model-based controllers capable of stabilizing nonlinear systems at nominal operating conditions.

However, in many applications, merely stabilizing the closed-loop system is not sufficient; achieving satisfactory performance is also crucial, often necessitating the integration of additional control loops. In Nonlinear Optimal Control (NOC), performance requirements are typically encoded in the shape of the cost function that the control policy strives to minimize. Consequently, it is beneficial to develop NOC algorithms that accommodate general nonlinear costs to enable sophisticated closed-loop behaviors, such as collision avoidance or waypoint tracking in swarms of robots.


In this paper, we tackle the following performance-boosting problem: given a discrete-time nonlinear system that is stable or has been pre-stabilized using a base controller, how can we enhance its performance during the transient --- that is, before the system settles into a steady state --- by employing general cost functions without compromising stability?

A first approach to designing performance-boosting regulators involves resorting to NOC methods with stability guarantees. Despite extensive research in this area \cite{sastry2013nonlinear}, the problem is fully understood only when the system dynamics are linear and the cost admits a convex reformulation. For nonlinear systems, traditional methods for addressing NOC include dynamic programming and the maximum principle \cite{bertsekas2011dynamic, pontryagin2018mathematical}. However, the computation of NOC policies through these methods often faces significant computational challenges \cite{pontryagin2018mathematical}. Furthermore, to ensure stability, stringent limitations must be imposed on the class of costs that can be utilized. An alternative approach to tackling performance-boosting is offered by receding-horizon control schemes, such as Nonlinear Model Predictive Control (NMPC) \cite{rawlings2017model}. These controllers are based on real-time optimization; a finite-horizon NOC problem is solved at each time instant to determine the control input. However, a significant limitation of NMPC is that the control policy can seldom be precomputed and stored in an explicit form, which makes NMPC inapplicable when the control platform lacks the computational resources necessary to solve mathematical programs in real-time. Moreover, similar to NOC, ensuring stability requires imposing strong limitations on the class of admissible cost functions \cite{rawlings2017model}.

More recently, Reinforcement Learning (RL) and Deep Neural Networks (DNNs) have emerged as powerful tools that enable agents to understand and optimally interact with complex environments and dynamical systems, e.g., \cite{sutton2018reinforcement, brunke2021safe}. 
Many RL approaches are based on minimizing arbitrary cost functions, calling for the use of broad sets of candidate nonlinear control policies. To this end, RL methods often employ families of policies that incorporate deep Neural Networks (NNs), due to their ability to model rich classes of nonlinear functions.
These capabilities have led to remarkable applications, such as four-legged robots navigating challenging terrains \cite{lee2020learning} and drones that can outperform humans in races \cite{song2021autonomous, kaufmann2023champion}. On the other hand, general methodologies for designing RL policies for nonlinear dynamical systems, while ensuring closed-loop stability, are currently scarce and may be limited by strong assumptions \cite{berkenkamp2018safe, zanon2020safe, jin2020stability}. 
As a result, so far the applicability of RL approaches has been mainly limited to systems that are not safety-critical.

Independent of their application in RL, NNs have been employed in model-based control since the 1990s for approximating nonlinear receding horizon policies \cite{parisiniRecedinghorizonRegulatorNonlinear1995,parisiniNonlinearStabilizationRecedinghorizon1998} or synthesizing nonlinear regulators from scratch \cite{levinControlNonlinearDynamical1996}.
Recent results on the design of provably stabilizing DNN control policies fall into two categories. The first one comprises constrained optimization approaches \cite{berkenkamp2018safe,gu2021recurrent,pauli2021offset} that ensure global or local stability by enforcing Lyapunov-like inequalities during optimization. However, conservative stability constraints can severely restrict the range of admissible policies or fail to produce a viable controller even when it exists.
Additionally, enforcing constraints such as linear matrix inequalities becomes a computational bottleneck in large-scale applications.

The second category embraces unconstrained optimization approaches, aiming to define classes of control policies with built-in stability guarantees \cite{wang2021learning,wang2022youla,furieri2022distributed}. These methods, which are similar to those developed in this paper, allow unconstrained optimization over finitely many parameters --- using, for instance, standard gradient descent techniques --- without sacrificing stability, regardless of the chosen parameter values.
Optimizing over sets of stabilizing policies has two main benefits. First, it completely decouples the stabilization problem from the choice of the cost being optimized. Second, it enables \textcolor{black}{\emph{stability by design}}, that is, the ability to guarantee closed-loop stability even if the policy optimization ends at a local minimum or is prematurely halted. However, these approaches are limited to discrete-time linear systems \cite{wang2021learning,wang2022youla} or to continuous-time systems in the port-Hamiltonian form \cite{furieri2022distributed}. While recent work surpasses the limitations above \cite{furieri2022neural,barbara2023learning}, in real-world applications, the knowledge about the system model is not perfect. The impact of modeling errors on the parametrizations of stable closed-loop maps for nonlinear systems has remained largely unexplored. 

\subsection{Contributions}

This paper explores approaches to solve performance-boosting problems in general discrete-time, time-varying systems. Specifically, we develop unconstrained optimization approaches based on classes of state-feedback policies that induce closed-loop dynamics described by \textcolor{black}{specific classes of} stable and deep NNs.

After formally stating the performance-boosting problem in Section~\ref{sec:problem_statement}, we present our first contribution, which provides a complete characterization of the class of stability-preserving controllers for stable systems. 
This result is presented in Section~\ref{sec:main} and reveals that an Internal Model Control (IMC) structure \cite{garcia1982internal, economou1986internal, bonassi2022recurrent} allows characterizing, without conservatism, the class of \emph{all} stability-preserving controllers, where the only free parameter is an $\mathcal{L}_p$ operator.
Our results hinge on adapting 
nonlinear variants of the Youla parametrization \cite{anantharam1984stabilization, fujimoto1998state, fujimoto2000characterization} to discrete-time systems in state-space with process noise, and revealing their connections with IMC schemes \cite{garcia1982internal, economou1986internal, bonassi2022recurrent} in this setup.

Further, we examine the relationship with the recently proposed nonlinear System Level Synthesis (SLS) framework developed in \cite{ho2020system}. In Section~\ref{sec:robustness}, our main contribution is that the proposed approach is compatible with scenarios where only an approximate system description is available, such as models identified from data or derived from simplified physical principles. 
Specifically, under a finite gain assumption on the model mismatch, stability can always be preserved by embedding a nominal system model and optimizing over nonlinear controllers with a sufficiently reduced gain on the free $\mathcal{L}_p$ parameter. Importantly, the method ensures vanishing conservatism \textcolor{black}{in the class of parametrized stabilizing policies} as the model uncertainty approaches zero. Additionally, by considering networks of interconnected subsystems, we demonstrate how the IMC structure of our controllers naturally lends itself to the development of distributed policies where the communication topology mirrors the subsystem couplings.

Finally, Section~\ref{sec:implementation} bridges the gap between theoretical developments and computations, showing how to use Recurrent Equilibrium Networks (RENs) \cite{kim2018standard, revay2023recurrent} to obtain a finite-dimensional parametrization of performance-boosting controllers that can include DNNs. 
The final part of the paper in Section~\ref{sec:numerical}  presents several simulations by considering coordination problems for mobile robots. Specifically, we show how, similarly to RL, the freedom in specifying the optimization cost allows designing NN controllers that can boost various forms of performance and safety, reaching beyond classical optimal control objectives consisting of the sum of stage-costs over time~\cite{bertsekas2011dynamic}.

This paper builds upon our initial work \cite{furieri2022neural} where we first derived the parametrization of all stabilizing controllers. However, unlike in \cite{furieri2022neural}, the IMC form of stabilizing controllers and the robustness analysis presented here are new. More specifically, the controllers in \cite{furieri2022neural} were based on the nonlinear SLS parametrization introduced in \cite{ho2020system}, while the controllers in this paper rely on a much more intuitive IMC formulation. Additionally, the main technical contributions about robustness with vanishing conservatism included in this paper are novel and not included in \cite{furieri2022neural}. Finally, the distributed control architectures and the majority of simulations presented in this work are not present in \cite{furieri2022neural}.

\smallskip

\subsection{Notation}
\label{sec:notation}
\textbf{Signals and operators:} 
The set of all sequences $\mathbf{x} = (x_0,x_1,x_2,\ldots)$, where $x_t \in \mathbb{R}^n$, $t\in \mathbb{N}$ , is denoted as $\ell^n$. 
Moreover,  $\mathbf{x}$ belongs to $\ell_p^n \subset \ell^n$ with $p \in \mathbb{N} \cup \infty$ if $\norm{\mathbf{x}}_p = \left(\sum_{t=0}^\infty |x_t|^p\right)^{\frac{1}{p}} < \infty$, where $|\cdot|$ denotes any vector norm. 
We say that $\xb \in \ell^n_\infty$ if $\operatorname{sup}_{t}|x_t|< \infty$. 
When clear from the context, we omit the superscript $n$ from $\ell^n$ and $\ell^n_p$. 
An operator $\mathbf{A}$ is said to be $\ell_p$-stable\footnote{We also say that the operator is \textit{stable}, for short, when the value of $p$ is clear from the context.}  if it is \emph{causal} and $\mathbf{A}(\mathbf{w}) \in \ell_{p}^m$ for all $\mathbf{w} \in \ell_{p}^n$. Equivalently, we  write $\mathbf{A} \in \mathcal{L}_{p}$. 
We say that an $\mathcal{L}_p$ operator $\mathbf{A}:\mathbf{w}\mapsto \mathbf{u}$  has finite $\mathcal{L}_p$-gain $\gamma(\mathbf{A})>0$ if  $\|\mathbf{u}\|_p\leq \gamma(\mathbf{A})\|\mathbf{w}\|_p$, for all $\mathbf{w}\in\ell_p^n$.

\textbf{Time-series:} We use the notation $x_{j:i}$ to refer to the truncation of $\mathbf{x}$ to the finite-dimensional vector $(x_i,x_{i+1},\ldots,x_{j})$. An operator $\mathbf{A}:\ell^n \rightarrow \ell^m$ is said to be \emph{causal} if $\mathbf{A}(\mathbf{x}) = (A_0(x_0),A_1(x_{1:0}),\ldots,A_t(x_{t:0}),\ldots)$. If in addition $A_t(x_{t:0}) = A_t(x_{t-1:0},0)$, then $\mathbf{A}$ is said to be strictly causal. Similarly,  we define $A_{j:i}(x_{j:0}) =  (A_i(x_{i:0}),A_{i+1}(x_{i+1:0}),\ldots,A_j(x_{j:0}))$. 
 For a matrix $M \in \mathbb{R}^{m \times n}$, $M\mathbf{x} = (Mx_0,Mx_1,\ldots) \in \ell^m$.

\textbf{Graph theory:} Given an undirected graph~$\cG=(\cV, \cE)$ described by the set of nodes $\cV=\{1,\ldots,N\}$ and the set of edges $\cE\subset \cV \times \cV$, we denote set of neighbors of node~$i$, including~$i$ itself by $\mathcal{N}_i = \{i\} \cup \{j\ |\ \{i,j\}\in\cE\} \subseteq \cV$. 
We denote with col$_{j\in\mathcal{V}}(v^{[j]})$ a vector which consists of the stacked
subvectors $v^{[j]}$ from $j=1$ to $j=N$ and with $v^{[\mathcal{N}_i]}$ a vector composed by the stacked subvectors $v^{[j]}$ of all neighbors of node~$i$,  
i.e., $v^{[\mathcal{N}_i]}=col_{j\in\mathcal{N}_i}(v^{[j]})$.  
For a signal $\mathbf{x} \in \ell^n$, where $x_t=col_{i\in\mathcal{V}}(x^{[i]}_t)$, $x_t^{[i]} \in \mathbb{R}^{n_i}$, and $n=\sum_{i=1}^N n_i$,
we denote with $\mathbf{x}^{[i]} \in \ell^{n_i}$ the sequence $\mathbf{x}^{[i]}= (x_0^{[i]},x_1^{[i]},\ldots)$. Similarly, 
we define 
$\mathbf{x}^{[\mathcal{N}_i]}= (x_0^{[\mathcal{N}_i]},x_1^{[\mathcal{N}_i]},\ldots)$.

\section{The Performance-boosting Problem}
\label{sec:problem_statement}

We consider nonlinear discrete-time time-varying systems
\begin{equation}
    \label{eq:system}
    x_{t} = f_t(x_{t-1:0},u_{t-1:0})+w_t\,,~~~t= 1,2,\ldots\,,
\end{equation}
where $x_t \in \mathbb{R}^n$ is the state vector, $u_t \in \mathbb{R}^m$ is the control input, $w_t \in \mathbb{R}^n$ stands for unknown process noise with $ w_0 =x_0 $, and $f_0=0$. The system model \eqref{eq:system} is very general. For instance, it can describe the dynamics of the error between the state of a nonlinear system and a reference trajectory in $\ell_p$.
In \emph{operator form}, system \eqref{eq:system} is equivalent to
\begin{equation}
\label{eq:operator_form}
    \mathbf{x} = \mathbf{F}(\mathbf{x},\mathbf{u}) + \mathbf{w}\,,\end{equation}
where $\mathbf{F}:\ell^n\times \ell^m \rightarrow \ell^n$ is the strictly causal operator such that $\mathbf{F}(\mathbf{x},\mathbf{u}) = (0,f_1(x_0,u_0),\ldots,f_t(x_{t-1:0},u_{t-1:0}),\ldots)$. 
Note that $\mathbf{w}=(x_0,w_1,\ldots)$ and $\mathbf{u}$ collects all data needed for defining the system evolution over an infinite horizon.
As an example, when the system~\eqref{eq:system} takes the Linear Time Invariant (LTI) form 
\begin{equation}
    \label{eq:linear_system}
    x_t = Ax_{t-1}+Bu_{t-1} + w_t \,,
\end{equation}
the model~\eqref{eq:operator_form} becomes
{
\setlength{\arraycolsep}{2pt}
\begin{equation*}
\begin{bmatrix}
x_0 \\  x_1 \\ x_2 \\ \vdots 
\end{bmatrix}
=
\begin{bmatrix}
0 & 0 & 0 & \cdots \\
A & 0 & 0 & \cdots \\
0 & A & 0 & \cdots \\
\vdots & \vdots & \vdots & \ddots
\end{bmatrix}
\begin{bmatrix}
x_0 \\  x_1 \\ x_2 \\ \vdots 
\end{bmatrix} + 
\begin{bmatrix}
0 & 0 & 0 & \cdots \\
B & 0 & 0 & \cdots \\
0 & B & 0 & \cdots \\
\vdots & \vdots & \vdots & \ddots
\end{bmatrix}
\begin{bmatrix}
u_0 \\ u_1 \\ u_2 \\ \vdots 
\end{bmatrix} + 
\begin{bmatrix}
x_0 \\  w_1 \\ w_2 \\ \vdots 
\end{bmatrix}\,.
\end{equation*}
}
We consider disturbances with support $\mathcal{W}_t\subseteq \mathbb{R}^n$ following a random vector distribution $\mathcal{D}_t$, that is, $w_t \in \mathcal{W}_t$ and $w_t \sim \mathcal{D}_t$ for every $t=0,1,\ldots$.
In order to control the behavior of system~\eqref{eq:system}, we consider nonlinear, state-feedback, time-varying control policies
\begin{equation}
    \label{eq:control}
    \mathbf{u} = \mathbf{K}(\mathbf{x}) = (K_0(x_0),K_1(x_{1:0}),\ldots, K_t(x_{t:0}),\ldots)\,,
\end{equation}
where $\mathbf{K}:\ell^n\ \rightarrow \ell^m$  is a \emph{causal} operator to be designed.  Note that the controller $\mathbf{K}$ can be dynamic, as $K_t$ can depend on the whole past history of the system state.  Since for each 
$\mathbf{w} \in \ell^n$  and $\mathbf{u} \in \ell^m$ 
the system~\eqref{eq:system} produces a unique state sequence $\mathbf{x} \in \ell^n$, equation \eqref{eq:operator_form} defines a unique transition operator
\begin{equation*}
    \Effe:(\mathbf{u},\mathbf{w})\mapsto \mathbf{x}\,,
\end{equation*} 
which provides an input-to-state model of system~\eqref{eq:system}. 
Similarly, for each $\mathbf{w} \in \ell^n$ the closed-loop system \eqref{eq:system}-\eqref{eq:control} produces unique trajectories. 
Hence, the closed-loop mapping $\mathbf{w} \mapsto(\mathbf{x},\mathbf{u})$ is well-defined. 
Specifically, for a system $\mathbf{F}$ and a controller $\mathbf{K}$, we denote the corresponding induced closed-loop operators $\mathbf{w} \mapsto\mathbf{x}$ and $\mathbf{w} \mapsto\mathbf{u}$ as $\Phix{F}{K}$ and $\Phiu{F}{K}$, respectively. 
Therefore, we have $\mathbf{x} = \Phix{F}{K}(\mathbf{w})$ and $\mathbf{u} = \Phiu{F}{K}(\mathbf{w})$ for all $\mathbf{w} \in \ell^n$. 
\begin{definition}
The closed-loop system \eqref{eq:system}-\eqref{eq:control} is $\ell_p$-stable if  $\Phiu{F}{K}$ and $\Phiu{F}{K}$ are in $\mathcal{L}_{p}$. 
\end{definition}

Our goal is to synthesize a control policy $\mathbf{K}$ solving the following problem.
\begin{problem}[Performance boosting]
\label{prob:boosting}
	Assume that $\Effe$ lies in $\mathcal{L}_p$. Find $\mathbf{K}$ solving the finite-horizon Nonlinear Optimal Control (NOC) problem

 \begin{subequations}
    \label{NOC:cost_and_stab}
	\begin{alignat}{3}
	&\min_{\mathbf{K}(\cdot)}&& \qquad \mathbb{E}_{w_{T:0}}\left[L(x_{T:0},u_{T:0})\right] \label{NOC:cost}\\
	&\operatorname{s.t.}~~ && x_{t} = f_t(x_{t-1:0},u_{t-1:0})+w_t\,, ~~w_0 = x_0\,, \nonumber\\
	&~~&&u_t = K_t(x_{t:0})\,,~~\forall t =0,1,\ldots\,,\nonumber \\
	&~~&&(\Phix{F}{K},\Phiu{F}{K}) \in \mathcal{L}_{p}\, \,, \label{NOC:stab} 
	\end{alignat}
 \end{subequations}
where $L(\cdot)$ defines any \textcolor{black}{piecewise differentiable lower bounded} loss over realized trajectories $x_{T:0}$ and $u_{T:0}$, and the expectation $\mathbb{E}_{w_{T:0}}[\cdot]$ removes the effect of disturbances $w_{T:0}$ on the realized values of the loss.\footnote{Another common choice is to use $\max_{w_{T:0} \in \mathcal{W}_{T:0}}[\cdot]$ instead of the expectation. Other useful choices  include $\operatorname{Var}_{w_{T:0}}[\cdot]$, $\operatorname{CVAR}_{w_{T:0}}[\cdot]$, and weighted combinations of all the above. In practice, one can approximate the chosen operator that removes the effect of disturbances from the cost by performing multiple experiments.} 
\end{problem}
The main feature of \eqref{NOC:cost_and_stab} is that the cost is optimized over the finite horizon $0,\ldots,T$, but under the strict requirement that the closed-loop system is stable when it evolves over $0,\ldots,+\infty$. In other words, the feedback controller must preserve stability of $\Effe$, and its role is to boost the performance of the system in the transient $0,\ldots, T$. As it will be clear in the sequel, \textcolor{black}{we consider iterative control design algorithms based on gradient descent that exclusively search within sets of controllers that are stability-preserving by design}. This guarantees closed-loop stability during the optimization of the policy parameters.
Note also that,  as it is standard in NOC, we do not expect gradient descent to find the globally optimal solution for any initialization — this is generally impossible for problems beyond Linear Quadratic Gaussian (LQG) control, which enjoy convexity of the cost and linearity of the optimal policies \cite{tang2021analysis,furieri2020first}. Furthermore, the expected value in  \eqref{NOC:cost} can seldom be computed\footnote{For instance because it is too costly or the distribution $\mathcal{D}$ is unknown.} and is approximated by using samples of $w_{T:0}$. \textcolor{black}{Our} design guarantees that, 
in spite of all these limitations, closed-loop stability is never lost.


\section{Parametrization of all Stability-preserving Controllers}
\label{sec:main}

\textcolor{black}{We show how to parametrize all and only the stability-preserving policies by using an IMC control architecture \cite{garcia1982internal,economou1986internal},} depending on an operator $\Emme$ that can be freely chosen in $\mathcal{L}_p$.
Specifically, the block diagram of the proposed control architecture is represented in Figure~\ref{fig:IMCscheme} and it includes a copy of the system dynamics, which is used for computing the estimate $\widehat \bmw$ of the disturbance $\bmw$.  \textcolor{black}{ A key advantage of the proposed IMC parametrization is its compatibility with recently proposed neural network dynamical system models such as those described in \cite{revay2023recurrent, kim2018standard}. As we discuss in Section~\ref{sec:implementation}, these models enable the learning of performance-boosting stabilizing controllers by optimizing a set of free parameters $\theta \in \mathbb{R}^d$, for instance, through simple gradient descent.} 
\begin{figure}
	\centering
    \includegraphics[width=0.99\linewidth]{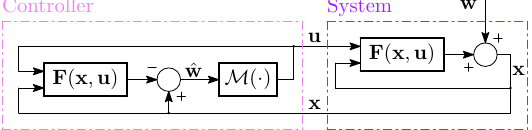}
	\caption{IMC architecture parametrizing of all stabilizing controllers in terms of one freely chosen operator $\Emme \in \mathcal{L}_p$.}
	\label{fig:IMCscheme}
\end{figure}
We are now in a position to introduce the main result.
\begin{theorem}
	\label{th:result_IMC} 
	Assume that the operator $\Effe$ is $\ell_p$-stable, i.e. $\bmx\in\ell_p$ if $(\bmw,\bmu)\in\ell_p$, and consider the evolution of \eqref{eq:operator_form} where $\mathbf{u}$ is chosen as 
 \begin{equation}
 \label{eq:input_M}
     \mathbf{u}=\Emme(\mathbf{x}-\mathbf{F}(\mathbf{x},\mathbf{u}))\,,
 \end{equation} 
 for a causal operator $\Emme:\ell^n \rightarrow \ell^m$. 
 Let $\mathbf{K}$ be the operator such that $\mathbf{u}=\mathbf{K}(\mathbf{x})$ is equivalent to \eqref{eq:input_M}.%
 \footnote{This operator always exists because $\mathbf{F}(\mathbf{x},\mathbf{u})$ is strictly causal. Hence $u_t$ depends on the inputs $u_{t-1:0}$ and can be computed recursively from past inputs and $x_{t:0}$ --- see formula \eqref{eq:controller_IMC}.}
 The following two statements hold true.
	\begin{enumerate}
		\item If $\Emme \in \mathcal{L}_{p}$, then the closed-loop system is $\ell_p$-stable.
		\item If there is  a causal policy $\bmC$ such that  $\Phix{\Fb}{\bmC},~\Phiu{\Fb}{\bmC} \in \mathcal{L}_p$, then 
			\begin{equation}
		\label{eq:choice_Emme}
		\Emme=\Phiu{\Fb}{\bmC}\,, 
		\end{equation}
		 gives $\bmK=\bmC$. 
	\end{enumerate}
\end{theorem}

\begin{proof}
	We prove $1)$.  For compactness, define $\hatbf{w}=\mathbf{x}-\mathbf{F}(\mathbf{x},\mathbf{u})$. 
    As highlighted in \cite{economou1986internal}, since there is no model mismatch between the plant $\Effe$ and the model $\mathbf{F}$ used to define $\hatbf{w}$, one has $\hatbf \bmw=\bmw$, hence opening the loop. More specifically, from Figure~\ref{fig:IMCscheme} and Equation~\eqref{eq:operator_form} one has
	\begin{equation}\label{eq:tweqw}
	   \hatbf \bmw = -\bmF(\bmx,\bmu)+ \bmF(\bmx,\bmu)+\bmw=\bmw\,.
	\end{equation}
	Therefore, by definition of the closed-loop maps,  one has $\Phiu{\Fb}{ \bmK}=\Emme$ and $\Phix{\Fb}{\bmK}(\bmw)=\bmF(\bmx,\Emme(\bmw))+\bmw$, $\forall \bmw \in \ell_p$.
	When $\bmw\in\ell_p$, one has $\Phiu{\Fb}{ \bmK}(\bmw)\in\ell_p$ because $\Emme\in\mathcal{L}_p$.   
    Moreover   $\Emme\in\mathcal{L}_p$ and $\Effe\in\mathcal{L}_p$
	imply that the operator $\bmw\mapsto \bmx$ defined by the composition of the operators $\bmw\mapsto(\Emme(\bmw),\bmw)$ and $\Effe$ is in  $\mathcal{L}_p$ as well. This is due to the property that the composition of operators in $\mathcal{L}_p$ is in $\mathcal{L}_p$.

	 We prove $2)$.   
  Set, for short, 
  $\Psix=\Phix{\Fb}{\bmC}$, 
  $\Psiu=\Phiu{\Fb}{\bmC}$, 
  $\Phixno =\Phix{\Fb}{ \bmK}$, and 
  $\Phiuno=\Phiu{\Fb}{ \bmK}$. By assumption, one has $\Emme=\Psiu$ and since $\Psiu\in\LL_p$ also $\Emme\in\LL_p$. By definition, $\Phiuno$ is the operator $\bmw\mapsto\bmu$ and, from \eqref{eq:tweqw} and Figure~\ref{fig:IMCscheme}, it coincides with $\Emme$. Hence \begin{equation}
\label{eq:PhiuisPsiu}
\Psiu=\Phiuno\,.
\end{equation}
It remains to prove that $\Phixno = \Psix$. Similar to \cite{furieri2022neural}, we proceed by induction. 
First, we show that $\Psi^x_{0}=\Upsilon^x_{0}$, where, as defined in Section~\ref{sec:introduction}.\ref{sec:notation}, $\Psi^x_{0}$ and $\Upsilon^x_{0}$ are the components of $\Psix$ and $\boldsymbol{\Upsilon}^{\mathbf{x}}$ at time zero.
Since $f_0=0$ and $w_0=x_0$, one has from \eqref{eq:system} that the closed-loop map $w_0\mapsto x_0$ is the identity, irrespectively of the controller. Therefore $\Upsilon_0^x = \Psi_0^x = I$.  
Assume now that, for a positive $j \in \mathbb{N}$ we have $\Upsilon^x_i = \Psi^x_i$ for all $0\leq i \leq j$. 
Since $(\Phixno,\Phiuno)$ and $(\Psix, \Psiu)$ are closed-loop maps, from \eqref{eq:operator_form} they verify
	\begin{equation}
	\label{eq:achievability_proof}
	\Upsilon^x_{j\hspace{-0.04cm}+\hspace{-0.04cm}1} \hspace{-0.1cm}=\hspace{-0.1cm} F_{j\hspace{-0.04cm}+\hspace{-0.04cm}1}(\Upsilon^x_{j:0},\Upsilon^u_{j:0})\hspace{-0.04cm}+\hspace{-0.01cm}I,\Psi^x_{j\hspace{-0.01cm}+\hspace{-0.04cm}1} \hspace{-0.1cm}=\hspace{-0.1cm} F_{j\hspace{-0.04cm}+\hspace{-0.01cm}1}(\Psi^x_{j:0},\Psi^u_{j:0})\hspace{-0.01cm}+\hspace{-0.04cm}I.
	\end{equation}
	But, from \eqref{eq:PhiuisPsiu}, one has $\Psi^u_{j:0}=\Upsilon^u_{j:0}$ and, by using the inductive assumption, one obtains $\Upsilon^x_{j\hspace{-0.04cm}+\hspace{-0.04cm}1}=\Psi^x_{j\hspace{-0.04cm}+\hspace{-0.04cm}1}$. This implies $\mathbf{K=C}$.
\end{proof}

Several comments are in order. 
First, Theorem~\ref{th:result_IMC} is about \textit{nominal stability} only as there is no model mismatch between the plant model and the one used in the controller. We analyze robust stability in Section~\ref{sec:robustness}. 
Second, it is well known that many IMC architectures are sufficient for preserving stability, both in the linear \cite{garcia1982internal} and the nonlinear \cite{economou1986internal} case.\footnote{Note, however, that  IMC  in \cite{economou1986internal} is developed in terms of continuous-time nonlinear input-output models, for which the effect of process noise is difficult to analyze. 
Moreover, the control objective is to track a reference signal to the plant output, which raises the problem of approximating inverses of nonlinear operators. In our work, we use instead discrete-time input-to-state models and analyze the closed-loop maps from process noise to control inputs and system states. Moreover, our goal is to solve optimal control rather than tracking problems.}  
It is also known that in the LTI setting, IMC is also necessary for preserving stability \cite{rivera1986internal} and provides an alternative to the Youla-Koucera parametrization \cite{zhou1998essentials}. 
\textcolor{black}{In this respect, Theorem~\ref{th:result_IMC}  provides a necessary condition for preserving stability also for nonlinear systems. This result is perhaps not surprising given that necessary and sufficient conditions for stabilizing wide classes of input-output nonlinear models, in the spirit of the Youla- Koucera parametrization, have been derived since the 80's \cite{anantharam1984stabilization, fujimoto2000characterization}. 
However, these controllers are not conceived in the  IMC form \cite{garcia1982internal, economou1986internal, bonassi2022recurrent} and they consider actuation and measurement disturbances, while our setup allows for the presence of process noise.
}

\textcolor{black}{The above insight is useful because the IMC structure facilitates the design and deployment of performance-boosting policies. First, IMC controllers are deployed using} the block-diagram structure shown in Figure~\ref{fig:IMCscheme}. In equation form, for a chosen operator $\Emme$, one simply computes the control input as follows:
\begin{subequations}
\label{eq:controller_IMC}
\begin{align}
    &\widehat{w}_t = x_t - f_t(x_{t-1:0},u_{t-1:0})\,, \label{eq:controller_IMC_1}\\
    &u_t = \mathcal{M}_t(\widehat{w}_{t:0})\,.\label{eq:controller_IMC_2}
\end{align}
\end{subequations}
\textcolor{black}{Second}, Theorem~\ref{th:result_IMC}  highlights that it is sufficient to search in the space of operators $\Emme\in\LL_p$ for describing all and only performance-boosting policies. While finding a parametrization of all operators $\Emme\in\LL_p$ might be prohibitive, we will show in Section \ref{sec:implementation} that one can use NNs for describing broad subsets of these operators. Moreover, the IMC structure lends itself to the development of policies that 
enjoy a distributed structure (see Section~\ref{sec:robustness}).

\subsubsection{The case of LTI systems with nonlinear costs}\label{subsub:linear}

\textcolor{black}{Consider the linear system \eqref{eq:linear_system} 
and let $z$ denote the forward time-shift operator.} When the system is asymptotically stable, the classical Youla  parametrization \cite{zhou1998essentials} states that all \emph{linear state-feedback} stabilizing control policies $\mathbf{u} = \mathbf{K}\xb$ can be written as
\begin{equation}
\label{eq:Youla_Linear}
    \mathbf{u} = \mathbf{Q}(z)\mathbf{x}-\frac{\mathbf{Q}(z)}{z}\left(A\mathbf{x}+B\mathbf{u}\right) \quad \mathbf{Q}(z) \in \mathcal{TF}_s\,,
\end{equation}
where $\mathbf{Q}(z)$ is the so-called Youla parameter. 
Here, $\mathcal{TF}_s$ denotes the set of stable transfer matrices --- that is, the set of matrices whose scalar entries are stable transfer functions.     
The class of linear control policies is globally optimal for standard LQG problems, and it allows optimizing over $\mathbf{Q} \in \mathcal{TF}_s$ using simple pole approximations and convex programming --- we refer to \cite{fisher2023approximation,fisher2022approximation} for state-of-the-art results.  However, nonlinear policies can be significantly more performing when the controller is distributed \cite{furieri2020sparsity}, or the cost function is nonlinear. As an immediate corollary of Theorem~\ref{th:result_IMC}, and in accordance with the core contribution of \cite{wang2022learning} \textcolor{black}{where the focus is on contracting closed-loops}, we have the following result for linear systems controlled by nonlinear policies.
\begin{corollary}
 Consider the linear system \eqref{eq:linear_system} and assume that it is asymptotically stable. 
 Then, all and only control policies 
 that make the closed-loop system $\ell_p$-stable
 are expressed as
 \begin{equation}
     \label{eq:all_policies}
     \mathbf{u} =  \Emme\left(\mathbf{x}-\frac{\left(A\mathbf{x}+B\mathbf{u}\right)}{z}\right)\,,
 \end{equation}
 where $\Emme \in \mathcal{L}_{p}$. 
 \end{corollary}
 
 \begin{proof}
    The proof follows from Theorem~\ref{th:result_IMC} upon realizing that the asymptotic stability of system \eqref{eq:linear_system} implies that the corresponding operator $\Effe$ is in $\mathcal{L}_p$, for any $p\geq 1$.
 \end{proof}
In conclusion, as expected, the linear Youla parametrization \eqref{eq:Youla_Linear} is a special case of the proposed parametrization \eqref{eq:all_policies} with $\Emme = \mathbf{Q}$ and $\mathbf{Q} \in \mathcal{TF}_s$.

\subsubsection{Relationships with \cite{furieri2022neural} and nonlinear SLS}
In \cite{furieri2022neural}, we provided a slight generalization of Theorem ~\ref{th:result_IMC} and the results in Section \ref{subsub:linear} by also considering unstable systems $\mathbf{x}=\tilde{\mathbf{ F}}(\mathbf{x},\mathbf{u})+\mathbf{w}$ for which a pre-stabilizing controller $\mathbf{K}'$ exists, so that the overall policy is
\begin{align}
    \label{eq:CL_1}
    \mathbf{u} = \mathbf{K}'(\mathbf{x}) + \Emme(\mathbf{\hatbf{w}})\,. 
\end{align}
By letting $\mathbf{F}(\mathbf{x},\mathbf{u})=\mathbf{F}(\mathbf{x},\mathbf{K}'(\xb)+\mathbf{u})$, and assuming that both $\Effe$ and $\mathbf{K}'$ lie in $\mathcal{L}_p$, Theorem~\ref{th:result_IMC} coincides with Theorem 2 in \cite{furieri2022neural}. However, when $\mathbf{K}'\not \in\mathcal{L}_p$, Theorem 2 in \cite{furieri2022neural} highlights that $\Emme \in \mathcal{L}_{p}$ may no longer be a necessary condition for closed-loop $\ell_p$-stability, while being still sufficient.

Moreover, as highlighted in \cite{furieri2022neural}, there is a deep link between  Theorem~\ref{th:result_IMC} and the SLS parametrization of stabilizing controllers \cite{wang2019system, ho2020system}.  The idea behind the SLS approach \cite{wang2019system, ho2020system} is to circumvent the difficulty of characterizing stabilizing controllers, by instead directly designing stable closed-loop maps. Let us define the set of all \emph{achievable} closed-loop maps for system $\mathbf{F}$ as
\begin{equation}
    \label{eq:Phi_ach}
    \mathcal{CL}[\mathbf{F}] = \{(\Phix{F}{K},\Phiu{F}{K})~|~\mathbf{K}\text{ is causal}\}\,,
\end{equation}
and the set of all \emph{achievable and stable} closed-loop maps as
\begin{equation}
    \label{eq:Phi_ach_stable}
    \mathcal{CL}_p[\Fb] = \{(\Psix,\Psiu) \in \mathcal{CL}[\Fb]~|~(\Psix,\Psiu) \in \mathcal{L}_{p}\}\,.
\end{equation}
Note that, if $(\Psix,\Psiu) \in \mathcal{CL}_p[\mathbf{F}]$, then $\mathbf{x} = \Psix(\mathbf{w}) \in \ell_{p}^n$ and $\mathbf{u} = \Psiu(\mathbf{w})\in \ell_{p}^m$ for all $\wb \in \ell_{p}^n$. Based on Theorem~III.3 of \cite{ho2020system}, and adding the requirement that the closed-loop maps must belong to $\mathcal{L}_p$, we summarize the main SLS result for nonlinear discrete-time systems.
\begin{theorem}[\emph{Nonlinear SLS parametrization \cite{ho2020system}}]
\label{th:SLS}
The following two statements hold true.
\begin{enumerate}
    \item The set $\mathcal{CL}_p[\Fb]$ of all achievable and stable closed-loop responses admits the following characterization:
    \begin{subequations}
    \begin{align}
        \mathcal{CL}_p[\Fb] = \{&(\Psix,\Psiu)|~~(\Psix,\Psiu)\text{ are causal}\,,\label{SLS_equality0}\\
        &\Psix = \mathbf{F}(\Psix,\Psiu)+\mathbf{I}\,,\label{SLS_equality1}\\
        &(\Psix,\Psiu) \in \mathcal{L}_{p} \label{SLS_equality2} \}\,.
    \end{align}
    \end{subequations}
\item For any $(\Psix,\Psiu) \in \mathcal{CL}_p[\Fb]$, the operator $\Psix$ is invertible and the causal controller 
\begin{equation}
   \mathbf{u} = \mathbf{K}(\mathbf{x}) = \Psiu\left((\Psix)^{-1}(\mathbf{x})\right)\label{eq:equivalent_representation}\,,
\end{equation}
 is the only one that achieves the stable closed-loop responses $(\Psix,\Psiu)$.

\end{enumerate}
\end{theorem}

Theorem~\ref{th:SLS} clarifies that any policy $\mathbf{K}(\xb)$ achieving  $\ell_p$-stable closed-loop maps can be described in terms of two causal operators $(\Psix,\Psiu) \in \mathcal{L}_{p}$ complying with the nonlinear functional equality \eqref{SLS_equality1}. Therefore, the NOC problem admits an equivalent Nonlinear SLS (N-SLS) formulation:
\begin{alignat}{3}
\operatorname{N-SLS:}~~&\min_{(\Psix,\Psiu)}&& \quad \mathbb{E}_{w_{T:0}}\left[L(x_{T:0},u_{T:0})\right]
\tag{$\star$}\\
&~~~\operatorname{s.t.}~~ && \quad x_t = \Psi^x_t(w_{t:0})\,,~~~u_t = \Psi^u_t(w_{t:0})\,,\nonumber\\
&~~&&\quad (\Psix,\Psiu)\in \mathcal{CL}_p[\Fb]\,,  t = 0,1,\ldots  \nonumber
\end{alignat}

\noindent According to Theorem~\ref{th:SLS}, the constraint $(\Psix,\Psiu)\in \mathcal{CL}_p[\Fb]$ is equivalent to requiring that $(\Psix,\Psiu)$ are causal and verify \eqref{SLS_equality1}-\eqref{SLS_equality2}. 
The constraint \eqref{SLS_equality1} simply defines the operator $\Psix$ in terms of $\Psiu$ and it can be computed explicitly because $\mathbf{F}$ is strictly causal. The main challenge is to comply with \eqref{SLS_equality2}. Indeed, it is hard to generate $\Psiu \in \mathcal{L}_{p}$ such that the corresponding $\Psix$ satisfies $\Psix \in \mathcal{L}_{p}$. The paper \cite{ho2020system} suggests directly searching over $\ell_p$-stable operators $(\Psix,\Psiu)$ and abandoning the goal of complying with \eqref{SLS_equality1} exactly.  One can then study robust stability when \eqref{SLS_equality1} only holds approximately as per Theorem IV.2 in \cite{ho2020system}. However, with the exception of polynomial systems \cite{conger2022nonlinear}, this way of proceeding may result in conservative control policies or fail to produce a stabilizing controller. Instead, for the case of stable or pre-stabilized systems,  Theorem \ref{th:result_IMC} can be seen as a way of parametrizing all stabilizing controllers that circumvents completely the problem of fulfilling \eqref{SLS_equality1}-\eqref{SLS_equality2}.

\section{Beyond Closed-loop Stability: Handling Model Uncertainty and Distributed Architectures}
\label{sec:robustness}

This section tackles the performance boosting problem (Problem~\ref{prob:boosting}) under more intricate real-world constraints beyond just closed-loop stability. Firstly, Theorem~\ref{th:result_IMC} suffers from requiring perfect plant knowledge for controller design. In reality, ensuring closed-loop stability despite an imperfect model is crucial. Secondly, control policies in large-scale applications like power grids and traffic systems are inherently distributed. This means they rely solely on local sensor data and communication, posing significant challenges to achieving network-level robustness and stability.


\subsection{Robustness against model-mismatch}

Let us denote the nominal model available for design as ${\hatbf{F}}(\mathbf{x},\mathbf{u})$ and the real unknown plant as 
\begin{equation}
    \label{eq:mismatch_plant}
    \mathbf{F}(\mathbf{x},\mathbf{u}) = \hatbf{F}(\bmx,\bmu) + \bm{\Delta}(\bmx,\bmu)\,,
\end{equation}
where $\bm{\Delta}$ is a strictly causal operator representing the model mismatch. 
Let $\delta_t(x_{t-1:0},u_{t-1:0})$ be the time representation of the mismatch operator $ \mathbf{\Delta}$. Since for each sequence of disturbances $\mathbf{w} \in \ell^n$ and inputs $\mathbf{u} \in \ell^m$ the dynamics represented by \eqref{eq:system} with $f_t(x_{t-1:0},u_{t-1:0})$ replaced by $\widehat{f}_t(x_{t-1:0},u_{t-1:0})+\delta_t(x_{t-1:0},u_{t-1:0})$  produces a unique state sequence $\mathbf{x} \in \ell^n$, the equation 
    \begin{equation}
\label{eq:operator_form_perturbed}
    \mathbf{x} = \mathbf{F}(\mathbf{x},\mathbf{u}) + \mathbf{w}\,,
\end{equation}
defines again a unique transition operator $\Effe:(\mathbf{u},\mathbf{w})\mapsto \mathbf{x}$, which provides an input-to-state model of the perturbed system.

Here, we show that when $\bm{\Delta}$ can be described by an $\mathcal{L}_p$ operator with finite gain, we can always design operators $\Emme$ with sufficiently small $\mathcal{L}_p$-gain that stabilize the real closed-loop system. More specifically, letting $\gamma_{\bm{\Delta}}$ be the maximum $\mathcal{L}_p$-gain of the model mismatch $\bm{\Delta}$, it is possible to design controllers $\mathbf{K}$ that comply with the following robust version of the stability constraint \eqref{NOC:stab}:
\begin{equation}
\label{eq:robust_stability}
    (\bm{\Phi}^{*}[\widehat{\bmF}+\bm{\Delta},\mathbf{K}]) \in \mathcal{L}_{p}\,, ~ 
    *\in\{\bmx,\bmu \} \,, ~
    \forall \bm{\Delta}|~\gamma(\bm{\Delta}) \leq \gamma_{\bm{\Delta}}\,. 
\end{equation}
This result, which is given in the next theorem, 
refers to the control scheme in Figure~\ref{fig:small_gain}. 

\begin{figure}
	\centering
    \includegraphics[width=0.8\linewidth]{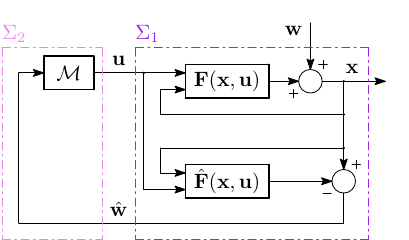}
	\caption{The closed-loop system when the nominal model ${\hatbf{F}}(\mathbf{x},\mathbf{u})$  used in the IMC controller and the real plant $
    \mathbf{F}(\mathbf{x},\mathbf{u})={\hatbf{F}}(\mathbf{x},\mathbf{u}) + {\mathbf{\Delta}}(\mathbf{x},\mathbf{u})$ differ by the perturbation ${\mathbf{\Delta}}\in\mathcal{L}_p$. Compared to Figure \ref{fig:IMCscheme} the blocks have been rearranged to highlight the subsystems used in the small-gain argument adopted in the proof of Theorem~\ref{th:result_robust}.}
\label{fig:small_gain}
\end{figure}

\begin{theorem}
	\label{th:result_robust} 
	Assume that the mismatch operator $\bm{\Delta}$ in \eqref{eq:mismatch_plant} has finite $\mathcal{L}_p$-gain $\gamma(\mathbf{\Delta})$. Furthermore, assume that
 the operator $\Effe$  has finite $\mathcal{L}_p$-gain $\gamma(\mathbf{\Effe})$.  Then, for any $\Emme$ such that 
 \begin{equation}
 \label{eq:condition_robustness}
 \gamma(\bmcalM)<\gamma(\mathbf{\Delta})^{-1}(\gamma(\Effe)+1)^{-1}\,,
 \end{equation}
 the control policy given by
 \begin{subequations}
 \label{eq:robust_policy}
 \begin{align}
 &\widehat{w}_t = x_t - \widehat{f}_t(x_{t-1:0},u_{t-1:0})\,,\label{eq:robust_policy_1}\\
 &u_t = \mathcal{M}_t(\widehat{w}_{t:0})\,,\label{eq:robust_policy_2}
 \end{align}
 \end{subequations}
 stabilizes the closed-loop system. 
 \end{theorem}
\begin{proof}
	We first show that operators $\mathbf{F}$ and $\Effe$ verify
 \begin{equation}
     \label{eq_F_adn_Effe}
     \mathbf{F}(\Effe(\mathbf{u},\mathbf{w}),\mathbf{u})=\Effe(\mathbf{u},\mathbf{w})-\mathbf{w}\,.
 \end{equation}
	This follows by substituting $ \mathbf{x}=\Effe(\mathbf{u},\mathbf{w})$  in \eqref{eq:operator_form_perturbed}. We now compute the $\mathcal{L}_p$-gain of the operator $\Sigma_1:(\mathbf{u},\mathbf{w})\mapsto\hatbf{w}$ in the right frame of Figure~\ref{fig:small_gain}:
 \begin{align}
    \hatbf{w} &= \Effe(\mathbf{u},\mathbf{w})-\hatbf{F}(\Effe(\mathbf{u},\mathbf{w}),\mathbf{u}) \nonumber \\
    &=\mathbf{F}(\Effe(\mathbf{u},\mathbf{w}),\mathbf{u})-\hatbf{F}(\Effe(\mathbf{u},\mathbf{w}),\mathbf{u})+\mathbf{w} \nonumber\\
    &= \mathbf{\Delta}(\Effe(\mathbf{u},\mathbf{w}),\mathbf{u})+ \mathbf{w} \,,\label{eq:u_to_twDelta}
\end{align}
 where the first equality follows from \eqref{eq_F_adn_Effe}. 
 Using the definition of $\mathcal{L}_p$-gain for the operator $\mathbf{y}=\mathbf{\Delta}(\mathbf{x},\mathbf{u})$ one has 
 $||\mathbf{y}||_p \leq \gamma(\mathbf{\Delta})(||\mathbf{x}||_p+||\mathbf{u}||_p)$, 
 and, by using \eqref{eq:u_to_twDelta} and $\mathbf{u} = \Emme(\hatbf{w})$, one obtains\footnote{For improving the clarity of the proof, from here onwards, we omit the subscript $p$ of the signal norms.}
  \begin{align*}
&||\hatbf{w}||\leq \gamma(\mathbf{\Delta})(||\Effe(\mathbf{u},\mathbf{w})||+||\mathbf{u}||)+||\mathbf{w}|| \nonumber \\
&\leq \gamma(\mathbf{\Delta})(\gamma(\Effe)||\mathbf{w}||+\gamma(\Effe) ||\mathbf{u}||+||\mathbf{u}||)+||\mathbf{w}|| \nonumber \\
&\leq (\gamma(\mathbf{\Delta})\gamma(\Effe)+1)||\mathbf{w}||+\gamma(\mathbf{\Delta})(\gamma(\Effe)+1) \gamma(\Emme)||\hatbf{w}||\,.
\end{align*}
\textcolor{black}{By gathering all the terms involving $||\hatbf{w}||$ to the left-hand side we obtain
\begin{equation*}
    (1-\gamma(\bm{\Delta}) \gamma(\Emme)\left(\gamma(\Effe)+1\right))||\hatbf{w}||\leq (\gamma(\bm{\Delta})\gamma(\Effe)+1)||\mathbf{w}||\,.
\end{equation*}
Since \eqref{eq:condition_robustness} holds, we have that $1-\gamma(\bm{\Delta})\gamma(\Emme)(\gamma(\Effe)+1)>0$, and hence}
\begin{equation}
\label{eq:map_w->w_hat}
    ||\hatbf{w}|| \leq \left(\frac{\gamma(\bm{\Delta})\gamma(\Effe)+1}{1-\gamma(\bm{\Delta}) \gamma(\Emme)\left(\gamma(\Effe)+1\right)} \right)||\wb||\,.
\end{equation}

Next, we plug the upper bound \eqref{eq:map_w->w_hat} into the inequality $||\mathbf{u}||\leq\gamma(\Emme)||\hatbf{w}||$ to obtain
\begin{equation}
    \label{eq:robust_loop_u}
    ||\mathbf{u}||\leq \left(\frac{\gamma(\Emme)\left(\gamma(\bm{\Delta})\gamma(\Effe)+1\right)}{1-\gamma(\bm{\Delta}) \gamma(\Emme)(\gamma(\Effe)+1)} \right)||\mathbf{w}||\,,
\end{equation}
and subsequently, we plug \eqref{eq:robust_loop_u} into 
the inequality $||\mathbf{x}||\leq \gamma(\Effe)(||\mathbf{u}||+||\mathbf{w}||)$ to obtain
\begin{equation}
    \label{eq:robust_loop_x}
    ||\mathbf{x}||\leq \left(\gamma(\Effe)\frac{1+\gamma(\Emme)\left(1-\gamma(\bm{\Delta})\right)}{1-\gamma(\bm{\Delta}) \gamma(\Emme)(\gamma(\Effe)+1)}\right)||\wb||\,.
\end{equation}
The last step is to verify that the maps $\wb\rightarrow\xb$ and $\wb\rightarrow\ub$ have a finite $\mathcal{L}_p$-gain. This is done by checking that the gains in \eqref{eq:robust_loop_u} and \eqref{eq:robust_loop_x} are positive values when the gain of $\Emme$ is sufficiently small. \textcolor{black}{Since  \eqref{eq:condition_robustness} holds}, the denominator in \eqref{eq:robust_loop_u} is positive. Since the numerator of \eqref{eq:robust_loop_u} is always positive, we conclude that the map $\wb\rightarrow \ub$ has an $\mathcal{L}_p$-gain. Similarly for \eqref{eq:robust_loop_x}, since \eqref{eq:condition_robustness} implies that $\gamma(\Emme) \gamma(\bm{\Delta})<1$, we have that both numerator and denominator are positive. This implies that the map $\wb\rightarrow \xb$ has an $\mathcal{L}_p$-gain, as desired.
\end{proof}

The robustness condition \eqref{eq:condition_robustness} highlights a trade-off between ($i$) the degree of tolerable uncertainty in the mismatch between nominal and real dynamics, and ($ii$) the extent of the set of stabilizing control policies that we are permitted to optimize over. Specifically, \eqref{eq:condition_robustness} ensures that, for any model mismatch $\bm{\Delta} \in \mathcal{L}_p$, there always exists a range of admissible gains for $\Emme$ such that the closed-loop is stable. This enables one to freely learn over all appropriately gain-bounded operators.  \textcolor{black}{Further note that  Theorem~\ref{th:result_robust} is not conservative when 
$\bm{\Delta}=0$ 
--- this is unlike the classical application of the small-gain theorem \cite{zames1966input} which would enforce that $\gamma(\mathbf{K})<(\gamma(\Effe))^{-1}$ even when $\bm{\Delta}=0$. 
Indeed, when the model is fully known, the right-hand side of \eqref{eq:condition_robustness} diverges to infinity, allowing the gain of $\Emme$ to be any finite value, although without imposing an upper bound, and therefore recovering the completeness result of Theorem~\ref{th:result_IMC}.} \textcolor{black}{Last, we remark that the relationships \eqref{eq:robust_loop_u} and \eqref{eq:robust_loop_x} formally quantify the extent to which the model mismatch can deteriorate the amplification of disturbances on the closed loop trajectories $(\mathbf{x},\mathbf{u})$ for the system, for a given policy. However, it remains open how much the model uncertainty deteriorates the performance of the \emph{optimal} policy. Such questions have only been rigorously answered for the linear-quadratic case, see, for instance, \cite{dean2017sample,zheng2021sample}.} 
%
\begin{remark}[Robust stability of nonlinear SLS]
    The authors of \cite{ho2020system} characterize robust stability of nonlinear SLS against mismatch 
    in satisfying the achievability constraint \eqref{SLS_equality1}. Specifically, \cite{ho2020system} focuses on the scenario where the control policy is a mapping $\xb\rightarrow\ub$ in the form 
    \begin{align}
        \tilde{\wb}&= \xb - (\Psix-\mathbf{I}) \tilde{\wb}\,, \label{eq:ho_control_1}\\
        \ub&= \Psiu(\tilde{\wb})\,,
    \end{align}
    \textcolor{black}{where $\tilde{\wb}$ represent the internal state of the controller,} for some $(\Psix,\Psiu) \in \mathcal{L}_p$ which are not assumed to perfectly comply with  \eqref{SLS_equality1}. Accordingly, the authors define a mismatch operator 
   \begin{equation}
   \label{eq:mismatch_operator_SLS}
       \bm{\Xi} = \mathbf{F}(\Psix,\Psiu)+\mathbf{I}-\Psix \,.
   \end{equation}
    Then, Theorem~IV.2 of \cite{ho2020system} proves closed-loop stability as long as $\gamma\left(\bm{\Xi}\right)<1$. Since $\bm{\Xi}$ measures the degree of violation of the achievability constraint rather than the degree of model uncertainty, a robust stability analysis based on verifying $\gamma(\bm{\Xi})<1$ tailored to the case $\mathbf{F}=\hatbf{F}+\bm{\Delta}$ may not be straightforward, and it is not attempted in \cite{ho2020system}. For this case, instead, Theorem~\ref{th:result_robust} provides an upper bound on the admissible gains for $\Emme$; this is achieved by 
    exploiting the IMC structure of the policy
    \eqref{eq:robust_policy}, and bounding the effect of model uncertainty on the closed-loop map for the ground-truth system.
\end{remark}

\subsection{Distributed controllers for large-scale plants}
\label{sec:distributed}
When dealing with large-scale cyber-physical systems, one may consider that 
the plant \eqref{eq:system} is composed of a network of $N$ dynamically interconnected nonlinear subsystems. To model this scenario,
we introduce an undirected coupling graph~$\mathcal{G}=(\mathcal{V}, \mathcal{E})$, where the nodes $\mathcal{V}=\{1, \dots, N\}$ represent the subsystems in the network, and the set of edges~$\mathcal{E}$ encode pairs of subsystems~$\{i,j\}$ that are dynamically interconnected through state variables. Specifically,  the dynamics of each subsystem $i\in \mathcal{V}$ is 
\begin{equation}
    x_t^{[i]}=f_t^{[i]}(x^{[\mathcal{N}_i]}_{t-1:0},u^{[i]}_{t-1:0})+w^{[i]}_t,  
    \ \ \    t=1,2,\ldots \label{eq:nonlinsyslocal}
\end{equation}
where state and input of each subsystem $i\in\mathcal{V}$ at time $t=1,2,\ldots$ are denoted by $x_t^{[i]}\in\mathbb{R}^{n_i}$ and $u_t^{[i]} \in \mathbb{R}^{m_i}$ respectively, and the initial state is
$x^{[i]}_0 \in \mathbb{R}^{n_i}$. In operator form we have
\begin{equation} \label{eq:sub-operator}
    \mathbf{x}^{[i]} =\mathbf{F}^{[i]}(\mathbf{x}^{[\mathcal{N}_i]},\mathbf{u}^{[i]})+\mathbf{w}^{[i]},
\end{equation}
where $\mathbf{F}^{[i]}:\ell^{n_{\mathcal{N}_i}}\times \ell^{m_i} \rightarrow \ell^{n_i}$. Note that, by stacking the subsystem dynamics in~\eqref{eq:nonlinsyslocal} together, we recover a system in the form \eqref{eq:system}, where $x_t=col_{i\in\mathcal{V}}(x_t^{[i]})\in\mathbb{R}^{n}$, $u_t=col_{i\in\mathcal{V}}(u_t^{[i]})\in\mathbb{R}^{m}$, and $w_t=col_{i\in\mathcal{V}}(w_t^{[i]})\in\mathbb{R}^{n}$.

When controlling networked systems in the form \eqref{eq:sub-operator}, a common scenario is that the local feedback controller $u_t^{[i]}$ can only access information made available by its neighbors according to a communication network with the same topology of  $\mathcal{G}$. 
This requirement translates into imposing the following additional constraint to the performance-boosting problem (Problem~\ref{prob:boosting}):
    \begin{equation}
   \mathbf{u}^{[i]} = \mathbf{K}^{[i]}(\mathbf{x}^{[\mathcal{N}_i]}),
   \quad \forall i\in \mathcal{V}\label{eq:distributed_constraint}\,.
    \end{equation}
The challenge becomes to parametrize only those stabilizing policies that are distributed according to \eqref{eq:distributed_constraint}. This can be achieved by exploiting the IMC controller architecture \eqref{eq:controller_IMC} in combination with the network sparsity of $\mathbf{F}$ highlighted in \eqref{eq:sub-operator}.  Let us consider, for example, the networked plant of Figure~\ref{fig:distributed_scheme}, where $\mathbf{u}^{[i]}$ depends on the local disturbance reconstructions $\hatbf{w}^{[i]}$ only, that is, $\mathbf{u}^{[i]}=\Emme^{[i]}(\hatbf{w}^{[i]})$.
In order to reconstruct $\hatbf{w}^{[1]}$, agent $i=1$ needs to evaluate the local dynamics $\mathbf{F}^{[1]}(\mathbf{x}^{[1]},\mathbf{x}^{[3]},\mathbf{u}^{[1]})$; this, in turns, requires a measurement of the state $\mathbf{x}^{[3]}$ over time. Repeating this reasoning for the agents $i=2$ and $i=3$, one obtains an overall control policy $\mathbf{K}(\mathbf{x})$ whose agent-wise components are computed relying on measurements from neighboring subsystems only, thus complying with \eqref{eq:distributed_constraint}. We formalize this reasoning in the next proposition.

\begin{proposition}
\label{prop:distributed}
Let graph $\mathcal{G} = (\mathcal{V},\mathcal{E})$ describe the topology of a plant $\mathbf{F}$ as per \eqref{eq:sub-operator}. Consider an IMC control policy \eqref{eq:controller_IMC} where the operator $\Emme \in \mathcal{L}_p$ is \emph{decentralized}, that is, $\Emme^{[i]}(\hatbf{w})=\Emme^{[i]}(\hatbf{w}^{[i]})$ for every agent $i \in \mathcal{V}$. Then, the closed-loop system is $\ell_p$-stable and the corresponding control policy 
$\mathbf{u} =\mathbf{K}(\mathbf{x})$ 
is distributed according to \eqref{eq:distributed_constraint}.
\end{proposition}
\begin{proof}
    Since $\Emme \in \mathcal{L}_p$, the closed-loop system is $\ell_p$-stable by Theorem~\ref{th:result_IMC}. 
    By \eqref{eq:sub-operator}, we have $\hatbf{w}^{[i]} = \mathbf{x}^{[i]}-\mathbf{F}^{[i]}(\mathbf{x}^{[\mathcal{N}_i]},\mathbf{u}^{[i]})$. Hence, agent $i$ only needs measurements of the neighboring states according to $\mathcal{G}$ and local past inputs, thus complying with \eqref{eq:distributed_constraint}.
\end{proof}

The result of Proposition~\ref{prop:distributed} can be extended to more complex cases. First, one can use local 
operators $\Emme^{[i]} \in \mathcal{L}_p$ that, besides $\hatbf{w}^{[i]}$, have access to disturbance reconstructions $\hatbf{w}^{[j]}$ or control variables $\ub^{[j]}$ computed at locations $j\neq i$. 
While these architectures can be beneficial, e.g. for counteracting disturbances affecting other subsystems before they propagate to the subsystem $i$ through coupling, they require additional communication channels $\{i,j\}$ if $j\not\in \mathcal{N}_i$. Moreover, one has to use local operators $\Emme^{[i]}$ guaranteeing that the whole operator $\Emme$ belongs to $\mathcal{L}_p$. 
To this purpose, in general, it is not enough that $\Emme^{[i]} \in \mathcal{L}_p$ because the dependency on $\hatbf{w}^{[j]}$ and $\mathbf{u}^{[j]}$ for $j\neq i$ can induce loop interconnections that can destabilize the closed-loop system. Classes of local operators  $\Emme^{[i]}$ yielding $\Emme \in \mathcal{L}_p$  have been proposed in \cite{massai2023unconstrained,saccani2024optimal} by using dissipativity theory.
\begin{figure}
	\centering
    \includegraphics[width=0.99\linewidth]{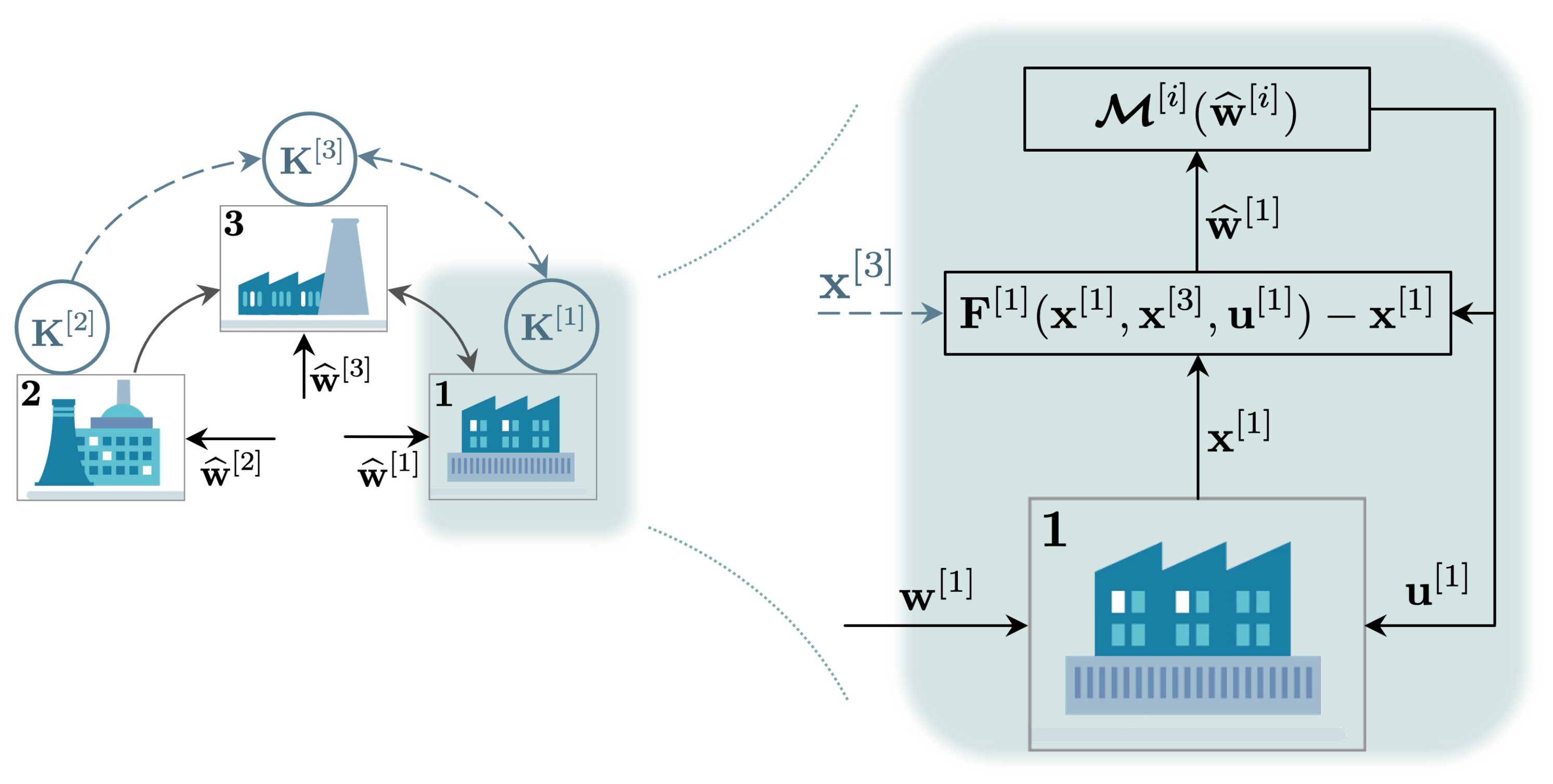}
	\caption{Example of networked dynamics \eqref{eq:sub-operator} and decentralized IMC controller for agent $i=1$.}
	\label{fig:distributed_scheme}
\end{figure}

\section{Learning to Boost Performance using Unconstrained Optimization}
\label{sec:implementation}

Leveraging the theoretical results of previous sections,  we reformulate the performance-boosting problem in a form that facilitates optimizing by automatic differentiation and unconstrained gradient descent. This enables the use of highly flexible cost functions for complex nonlinear optimal control tasks. By design, the proposed approach guarantees closed-loop stability throughout the optimization process. We assess the effectiveness of the proposed methodology in achieving optimal performance through numerical experiments, in Section~\ref{sec:numerical}.

\subsection{IMC-based reformulation of performance boosting}
\label{subsec:reformulation}
The main value of Theorem~\ref{th:result_IMC} is that it enables reformulating Problem~\ref{prob:boosting} as follows.

\medskip

\noindent\textbf{IMC reformulation of the performance-boosting problem:}
\begin{subequations}
\label{eq:Prob_reform}
\begin{alignat}{3}
	&\min_{\Emme \in \mathcal{L}_p}&& \qquad \mathbb{E}_{w_{T:0}}\left[L(x_{T:0},u_{T:0})\right] \label{eq:Prob_reform_1}\\
	&\operatorname{s.t.}~~&&x_t = f_t(x_{t-1:0},u_{t-1:0})+w_t, \quad x_0=w_0, \label{eq:Prob_reform_2}\\
    &~~&&u_t = \mathcal{M}_t({w}_{t:0})\,, \quad t=1,2,\ldots \,.\label{eq:Prob_reform_4}
\end{alignat}
\end{subequations}
Indeed, \eqref{eq:input_M} corresponds to \eqref{eq:Prob_reform_2}-\eqref{eq:Prob_reform_4}. If the exact dynamics $f_t$ in \eqref{eq:Prob_reform_2} is not known, it must be simply replaced by the nominal model $\widehat{f}_t$.

The reformulation \eqref{eq:Prob_reform} offers significant computational advantages as compared to 
Problem~\ref{prob:boosting}. In the classical linear quadratic case,%
\footnote{That is, when $f_t$ and $\Emme$ are linear and $L$ is quadratic positive definite.}
\eqref{eq:Prob_reform} becomes strongly convex in $\Emme$ --- enabling to use efficient convex optimization for finding a globally optimal solution \cite{youla1976modern,wang2019system, furieri2019input,zheng2022system,fisher2023approximation}. 
In the general nonlinear case, searching over nonlinear operators $\Emme \in \mathcal{L}_p$ remains significantly easier than tackling Problem~\ref{prob:boosting} directly. 
Indeed, the set $\mathcal{K}$ of controllers $\mathbf{K}(\cdot)$ complying with \eqref{NOC:stab} is, in general, difficult to parametrize. This is mainly because, given two stabilizing policies $\mathbf{K}_1, \mathbf{K}_2$, their convex combinations $\mathbf{K}_3=\gamma\mathbf{K}_1+(1-\gamma)\mathbf{K}_2$ with $\gamma \in [0,1]$ and their cascaded composition $\mathbf{K}_4 = \mathbf{K}_2(\bm{\Phi}^{\xb}[\bm{F},\mathbf{K}_1])$ do not result in stabilizing policies, in general; these issues are very well-known for the special case of linear systems \cite{fazel2018global,youla1976modern}. Hence, it is difficult to parameterize stabilizing policies, for instance, by composing or summing together base stabilizing operators. 
Instead, thanks to $\mathcal{L}_p$ being convex and closed under composition, there exist methods for parametrizing rich subsets of $\mathcal{L}_p$ through free parameters $\theta\in\mathbb{R}^d$, \textcolor{black}{where $d \in \mathbb{N}$ is the number of scalar parameters}, that is, to define operators $\Emme(\theta)$ such that 
\begin{equation}
    \label{eq:M:stable}
    \Emme(\theta) \in \mathcal{L}_p, \quad \forall \theta \in \mathbb{R}^d\,.
\end{equation}
This allows turning \eqref{eq:Prob_reform} into an unconstrained optimization problem over~$\theta \in \mathbb{R}^d$.

The last issue to be addressed is the computation of the average in \eqref{eq:Prob_reform_1} that, as noticed before, is generally intractable. This is usually circumvented by approximating the exact average with its empirical counterpart obtained using a set of samples $\{w_{T:0}^s\}_{s=1}^S$ drawn from the distribution $\mathcal{D}_{T:0}$. 
One then obtains the finite-dimensional optimization problem:
\begin{subequations}
\label{eq:learning}
\begin{alignat}{3}
	 \quad &\min_{\theta \in \mathbb{R}^d}&& \frac{1}{S}\sum_{s=1}^SL(x_{T:0}^s,u_{T:0}^s)\label{eq:learning_1}\\
	&\operatorname{s.t.}~~ && x_{t}^s = f_t(x_{t-1:0}^s,u_{t-1:0}^s)+w_t^s\,, ~~w_0^s = x_0^s\,,\label{eq:learning_2}\\
    &~~&&u_t^{\textcolor{black}{s}} = \mathcal{M}_t(\theta)(w_{t:0}^s)\,, \quad t=0,1,2,\ldots \,,\label{eq:learning_4}
	\end{alignat}
\end{subequations}
where $x_{T:0}^s$ and $u_{T:0}^s$ are the inputs and states obtained when the disturbance $w_{T:0}^s$ is applied. While in this work we only consider the empirical cost in the optimization problem~\eqref{eq:learning_1}, the closed-loop performance when faced with out-of-sample noise sequences is further investigated in~\cite{boroujeni2024pac}.

\textcolor{black}{Finally, we highlight that \eqref{eq:learning_2} and \eqref{eq:learning_4} can be seen as the equations of the layer $t$ of a neural network with $T$ layers. Specifically, we can interpret the layer $t$ of this neural network to have inputs $(x^s_{t-1:0},u^s_{t-1:0},w^s_{t:0})$ and outputs $(x^s_t,u^s_t)$.  Under this lens, the weights to be learned across all layers are the $\theta \in \mathbb{R}^d$ defining the control policy \eqref{eq:learning_4}.} 
When $\mathcal{M}_t$, for $t=0,1,\ldots$ is sufficiently smooth, the absence of constraints on $\theta$ enables the use of powerful packages, such as TensorFlow~\cite{tensorflow2015-whitepaper} and PyTorch~\cite{NEURIPS2019_9015}, leveraging
automatic differentiation and backpropagation for optimizing the controller through gradient descent.


\subsection{Free parameterizations of $\mathcal{L}_2$ subsets}
\label{subsec:RENs}

As highlighted in Section~\ref{sec:implementation}.\ref{subsec:reformulation}, the possibility of obtaining effective controllers by solving \eqref{eq:learning} critically depends on our ability to parametrize $\mathcal{L}_p$ operators. The main obstacle is that the space $\mathcal{L}_{p}$ is infinite-dimensional. Hence, for implementation, one usually restrict the search in subsets of $\mathcal{L}_{p}$ described by finitely many parameters. When linear systems are considered, 
one can search over Finite Impulse Response (FIR) transfer matrices $\mathbf{M} = \sum_{i=0}^N M[i]z^{-i} \in \mathcal{TF}_s$ 
and then optimize over the finitely many real matrices $M[i]$. Less and less conservative solutions can be obtained by increasing the FIR order $N$. However, the FIR approach limits the search to linear control policies. 

Recently, \cite{kim2018standard,revay2023recurrent,martinelli2023unconstrained} have proposed finite-dimensional  DNN approximations of \textcolor{black}{classes of} nonlinear $\mathcal{L}_{2}$ operators. In the sequel we briefly review the Recurrent Equilibrium Network (REN) models proposed in \cite{revay2023recurrent}.
An operator $\Emme:\ell^n\rightarrow \ell^m$ is a REN if the relationship $\mathbf{u} = \Emme(\hatbf{\wb})$ is recursively generated by the following dynamical system:
\begin{equation}
\label{eq:RENs}
        \begin{bmatrix}
		\xi_t \\  z_t \\ u_t
	\end{bmatrix}
	=
	\overbrace{
	\begin{bmatrix}
		A_{1} & B_1 & B_2 \\
		C_1 & D_{11} & D_{12} \\
		C_2 & D_{21} & D_{22}
	\end{bmatrix}
	}^{W}
	\begin{bmatrix}
		\xi_{t-1} \\
		\sigma(z_t)\\
		w_t
	\end{bmatrix}+
 	\overbrace{\begin{bmatrix}
		b_{x,t} \\
		b_{z,t}\\
		b_{w,t}
	\end{bmatrix}}^{b_t}
 \,,\quad \xi_{-1} = 0\,,
\end{equation}
where $\xi_t \in \mathbb{R}^{q}$, $v_t \in \mathbb{R}^r$, $b_{x,t},b_{z,t},b_{w,t}\in\ell_\infty$%
\footnote{This is slightly different from the original REN model, where these signals \cite{revay2023recurrent} are assumed to be constant.} 
and $\sigma:\mathbb{R} \rightarrow \mathbb{R}$ --- the activation function --- is applied element-wise. 
Further, $\sigma(\cdot)$ must be piecewise differentiable and with first derivatives restricted to the interval $[0,1]$. 
As noted in \cite{revay2023recurrent}, RENs subsume many existing DNN architectures. In general, RENs define \textit{deep equilibrium network} models \cite{bai2019deep} due to the implicit relationships defining $z_t$ in the second block row of \eqref{eq:RENs}. 
By restricting $D_{11}$ to be strictly lower-triangular, the value of $z_t$ can be computed explicitly, thus significantly speeding-up computations \cite{revay2023recurrent}. 
To give an example of the expressivity of \eqref{eq:RENs}, by suitably choosing the size and zero pattern of matrices in \eqref{eq:RENs}, RENs can provide nonlinear systems in the form
\begin{align*}
&  \xi_{t} =\hat A \xi_{t-1}+\hat B \,\text{NN}^{\xi} (\xi_{t-1},\widehat{w}_{t}) \\  
&u_t= \hat C\xi_t + \hat D \,\text{NN}^{u}(\xi_{t-1},\widehat{w}_{t})
\end{align*}
where $\hat A$, $\hat B$, $\hat C$, $\hat D$ are arbitrary matrices of suitable dimensions and $NN^\star$, $\star\in\{\xi,u\}$, are neural networks of depth $L$ given by the relations
\begin{align*}
&\tilde z_{0,t}^\star=[\xi_{t-1}^\top,\hat w_t^\top]^\top, \\
& \tilde z_{k+1,t}^\star=\sigma(W_k^\star \tilde z_{k,t}^\star+b_k^\star),\quad k=0,\ldots L-1
\end{align*}
where $W_k^\star$ and $b_k^\star$ are the layer weights and biases, respectively, and $\tilde z_{L,t}^\star$ is the NN output. 


For an arbitrary choice of $W$ and $b_t$, the map $\Emme$ induced by \eqref{eq:RENs} may not lie in $\mathcal{L}_2$. The work \cite{revay2023recurrent} provides an explicit smooth mapping $\Theta:\mathbb{R}^d \rightarrow \mathbb{R}^{(q+r+m) \times (q+r+n)}$ from unconstrained training parameters $\theta \in \mathbb{R}^d$ to a matrix $W=\Theta(\theta) \in \mathbb{R}^{(q+r+m) \times (q+r+n)}$ defining \eqref{eq:RENs}, with the property that the corresponding operator $\Emme(\theta)$ lies in $ \mathcal{L}_2$ by design when $b_t=0$.%
\footnote{Furthermore, RENs enjoy contractivity — although the theoretical results of this paper do not rely on this property.}
This approach can be easily generalized by including vectors $b_t$, $t=1,\ldots,T$ in the set of trainable parameters and assuming $b_t=0$ for $t>T$.

Recently, free parameterizations of continuous-time $\mathcal{L}_2$ operators through RENs and port-Hamiltonian systems have been also proposed in \cite{martinelli2023unconstrained} and \cite{zakwanNeuralDistributedControllers2024}, respectively.
\begin{remark}
The work \cite{wang2022learning} proves that RENs in the form \eqref{eq:RENs} are universal approximators of all contracting and Lipschitz operators when the parameters $(W,b)$ do not vary with time. To the best of the authors' knowledge, it is still unknown if the class of RENs in \(\mathcal{L}_2\), parametrized by \(W = \Theta(\theta)\) where \(W \in \mathbb{R}^{(q+r+m) \times (q+r+n)}\), can approximate any operator in \(\mathcal{L}_2\) arbitrarily well.
Our work motivates future research efforts to discover new parametrizations of operators in $\mathcal{L}_p$ with stronger and provable approximation capabilities.
\end{remark}

To conclude, we clarify that RENs can be directly embedded into the performance-boosting optimization problem \eqref{eq:learning_1}-\eqref{eq:learning_4}. This is obtained by substituting the input equation \eqref{eq:learning_4} with the recursions \eqref{eq:RENs}, where $W=\Theta(\theta)$ according to the mapping proposed in \cite{revay2023recurrent}.

\section{Numerical Experiments: the Magic of the Cost}
\label{sec:numerical}
In this section, we test the flexibility of performance boosting by considering cooperative robotics problems. Firstly, we validate the \textcolor{black}{guarantees} of the design approach by showing that closed-loop stability is preserved during and after training --- both when the system model is known and when it is uncertain. Secondly, we exploit the freedom in selecting the cost $L(x_{T:0},u_{T:0})$
to include appropriate terms aimed at promoting complex closed-loop behaviors. 

In all the examples, we consider two point-mass vehicles, each with position $p_t^{[i]} \in \mathbb{R}^2$ and velocity $q_t^{[i]} \in \mathbb{R}^2$, for $i=1,2$, subject to nonlinear drag forces (e.g., air or water resistance). The discrete-time model for vehicle $i$ is 
\begin{equation}
\label{eq:mechanical_system}
    \begin{bmatrix}
        p_{t}^{[i]}\\
        q_{t}^{[i]}
    \end{bmatrix}
    = 
    \begin{bmatrix}
        p_{t-1}^{[i]}\\
        q_{t-1}^{[i]}
    \end{bmatrix} 
    + T_s
    \begin{bmatrix}
        q_{t-1}^{[i]}\\
        (m^{[i]})^{-1}\left(-C(q_{t-1}^{[i]}) + F_{t-1}^{[i]}\right)\end{bmatrix}\,,
\end{equation}
where $m^{[i]}>0$ is the mass, $F_t^{[i]} \in \mathbb{R}^2$ denotes the force control input, $T_s>0$ is the sampling time and $C^{[i]}:\mathbb{R}^2\rightarrow \mathbb{R}^2$ is a \emph{drag function} given by
$C^{[i]}(s) = b_1^{[i]} s - b_2^{[i]} \tanh(s)$, for some $0<b_2^{[i]}<b_1^{[i]}$.
Each vehicle 
must reach a target position $\overline{p}^{[i]}\in \mathbb{R}^2$ with zero velocity in a stable way. This elementary goal can be achieved by using a base proportional controller
\begin{equation}
\label{eq:pre_controller}
    {F'}_{t}^{[i]} = {K'}^{[i]}(\bar p^{[i]}-p_{t}^{[i]})\,,
\end{equation}
with $K'^{[i]} = \operatorname{diag}(k_1^{[i]},k_2^{[i]})$ and $k_1^{[i]},k_2^{[i]}>0$. The overall dynamics $f_t(x_{t-1:0},u_{t-1:0})$ in \eqref{eq:system} is given by \eqref{eq:mechanical_system}-\eqref{eq:pre_controller} with 
\begin{equation}
\label{eq:control_input_additional}
F^{[i]}_t = F'^{[i]}_t + u_t^{[i]}\,,    
\end{equation}
where $x_t = (p_t^{[1]},q_t^{[1]},p_t^{[2]},q_t^{[2]})$ and 
$u_t = (u_t^{[1]},u_t^{[2]})$ is a performance-boosting control input to be designed.  
As per~\eqref{eq:system}, we consider additive disturbances affecting the system dynamics.
Thanks to the use of the prestabilizing controller \eqref{eq:pre_controller}, one can show that $\Effe(\bmu,\bmw)\in \mathcal{L}_2$.


The goal of the performance-boosting policy is to enforce additional desired behaviors, on top of stability, which are specified in each of the following subsections.  In all cases,  we parametrize the operator $\Emme(\theta) \in \mathcal{L}_2$ as a REN, see \eqref{eq:RENs}. 
Appendix~\ref{app:implementation} presents all the implementation details, such as parameter values and exact definitions of the cost functions.
Appendix~\ref{app:comparison} compares the performance of our methods and the corresponding guarantees with two related baseline approaches.
The code to reproduce our examples as well as 
various movies are available in our \href{https://github.com/DecodEPFL/performance-boosting_controllers.git}{Github repository}.\footnote{\url{https://github.com/DecodEPFL/performance-boosting_controllers.git}}

\subsection{Robust stability preservation during optimization}
\label{sec:example_stab}

We consider the scenario \texttt{mountains} in Figure~\ref{fig:corridor} where
each vehicle must reach the target position in a stable way while avoiding collisions between themselves and with two grey obstacles. 
Each agent is represented with a circle that indicates its radius for the collision avoidance specifications.
When using the base controller \eqref{eq:pre_controller}, the vehicles successfully achieve the target, however, they do so with poor performance since collisions are not avoided, as shown in Figure~\ref{fig:corridor}(a).


We select a loss $L(x_{T:0},u_{T:0})$ as the sum of stage costs $l(x_t,u_t)$, that is, $L(x_{T:0},u_{T:0}) = \sum_{t=0}^Tl(x_t,u_t)$ with
\begin{equation}\label{eq:loss_CA}
	l(x_t,u_t)= l_{traj}(x_t,u_t) + l_{ca}(x_t) + l_{obs}(x_t)\,,
\end{equation}
where 
$l_{traj}(x_t,u_t) 
= 
\begin{bmatrix} x_t^\mathsf{T} & u_t^\mathsf{T}\end{bmatrix} 
Q 
\begin{bmatrix} x_t^\mathsf{T} & u_t^\mathsf{T} \end{bmatrix}^\mathsf{T}$ 
with 
$Q\succeq 0$ penalizes the distance of agents from their targets and the control energy, $l_{ca}(x_t)$ and $l_{obs}(x_t)$ penalize collisions between agents and with obstacles, respectively. 

In order to train the performance-boosting controller, we solve~\eqref{eq:learning}, using a REN~\eqref{eq:RENs} of dimension $q= r = 8$.
The training data consists of a set of 100 initial positions, i.e.,
we set $w_0=((p^x_0)^{[1]},(p^y_0)^{[1]},0,0,(p^x_0)^{[2]},(p^y_0)^{[2]},0,0)$ and $w_t=0$, for $t>0$, where $p^x$ and $p^y$ denote the $x$ and $y$ coordinates of the vehicles in the Cartesian plane, respectively.
Initial positions are sampled from a Gaussian distribution around the nominal initial condition. 
Figure~\ref{fig:corridor}(b-c) shows the nominal and training initial conditions marked with `$\times$'  and `$\circ$', respectively, and three test trajectories after the training of the IMC controller. 
The trained control policies avoid collisions and achieve optimized trajectories thanks to minimizing \eqref{eq:loss_CA}.


\begin{figure}
	\centering
	\begin{minipage}{0.32\linewidth}
	    \includegraphics[width=\linewidth]{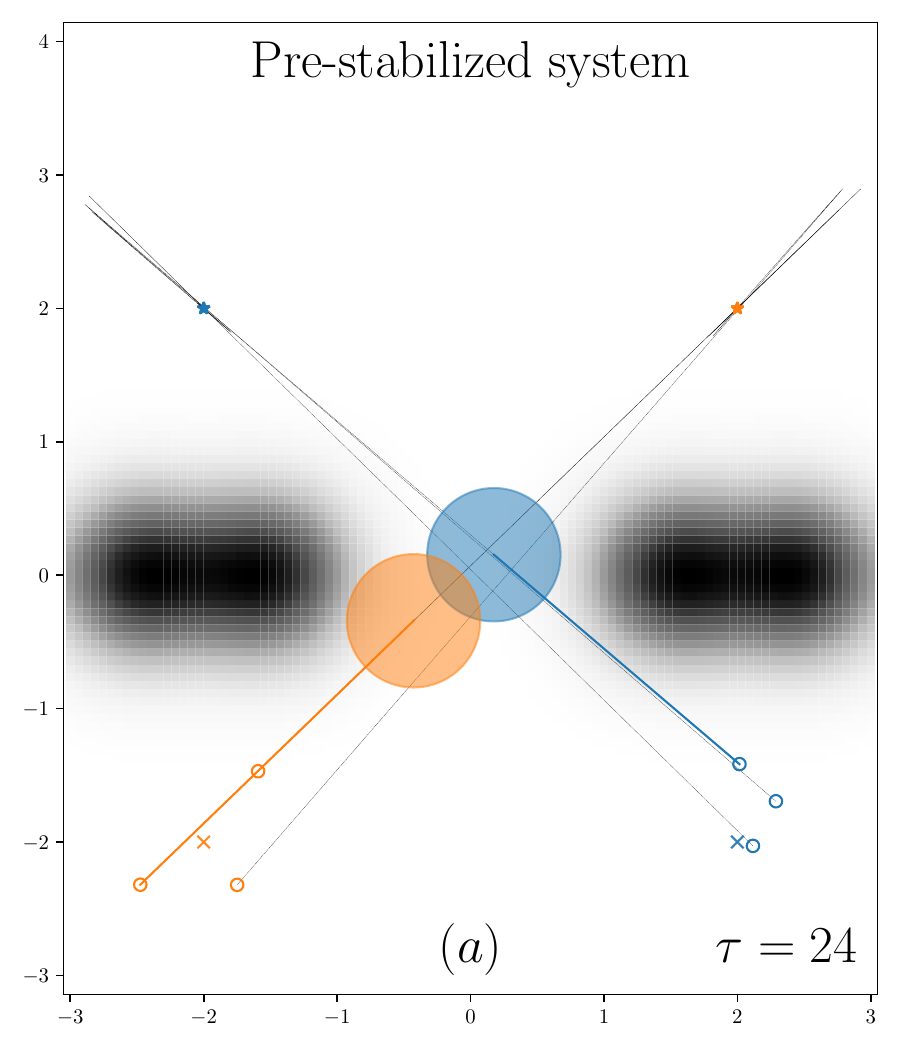}
	\end{minipage}%
	\begin{minipage}{0.32\linewidth}
	    \includegraphics[width=\linewidth]{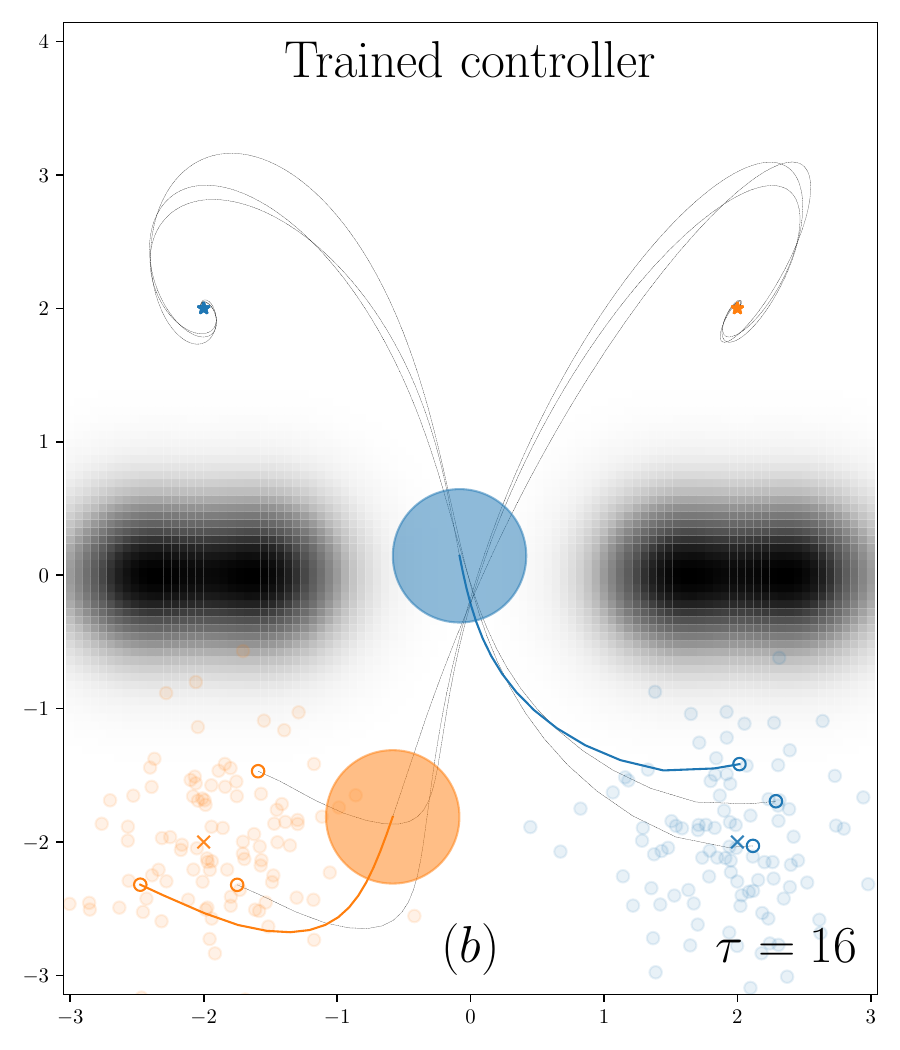}
	\end{minipage}%
	\begin{minipage}{0.32\linewidth}
	    \includegraphics[width=\linewidth]{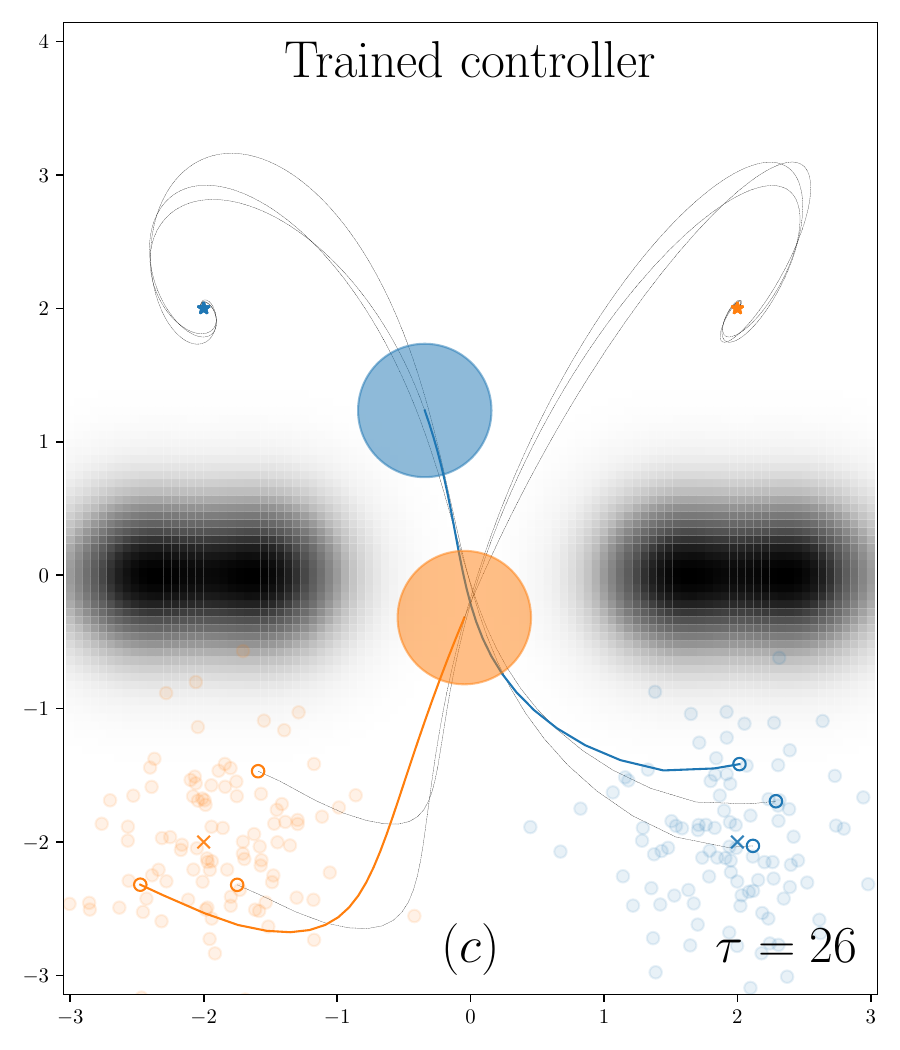}
	\end{minipage}
	\caption{\texttt{Mountains} --- Closed-loop trajectories before training (left) and after training (middle and right) over 100 randomly sampled initial conditions marked with $\circ$. Snapshots taken at time-instants $\tau$. Colored (gray) lines show the trajectories in $[0,\tau_i]$ ($[\tau_i,\infty)$). Colored balls (and their radius) represent the agents (and their size for collision avoidance).} 
	\label{fig:corridor}
\end{figure}

\subsubsection{Early stopping of the training}
We validate the \textcolor{black}{stability-by-design} property of our IMC control policies. 
We consider the scenario \texttt{mountains} as above but where the training process is interrupted before achieving a local minimum, as per the one in Figure~\ref{fig:corridor}.
In particular, we stop the optimization algorithm after 25\%, 50\%, and 75\% of the total number of epochs. The obtained trajectories are shown in Figure~\ref{fig:corridor_early_stopping}.
We observe that even if the performance is not optimized, closed-loop stability is always guaranteed.

\begin{figure}
	\centering
	\begin{minipage}{0.32\linewidth}
        \centering
	    \includegraphics[width=\linewidth]{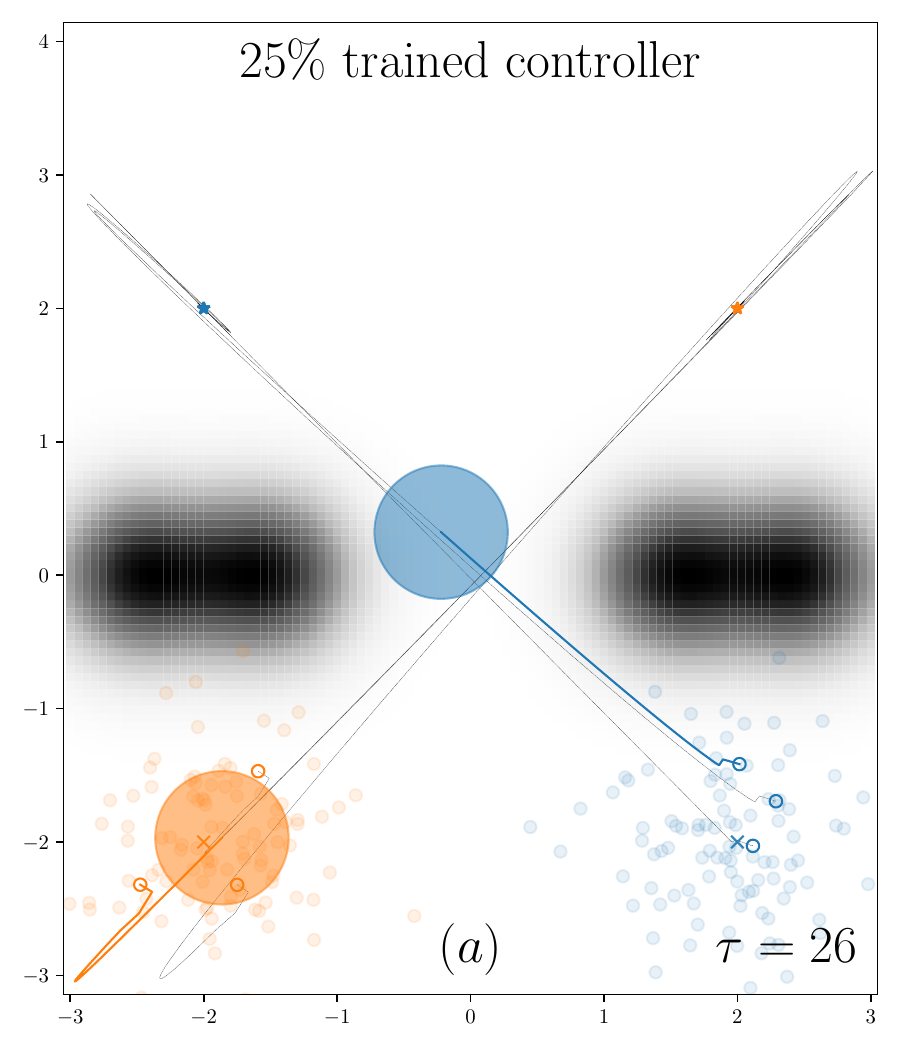}
	\end{minipage}%
	\begin{minipage}{0.32\linewidth}
	    \includegraphics[width=\linewidth]{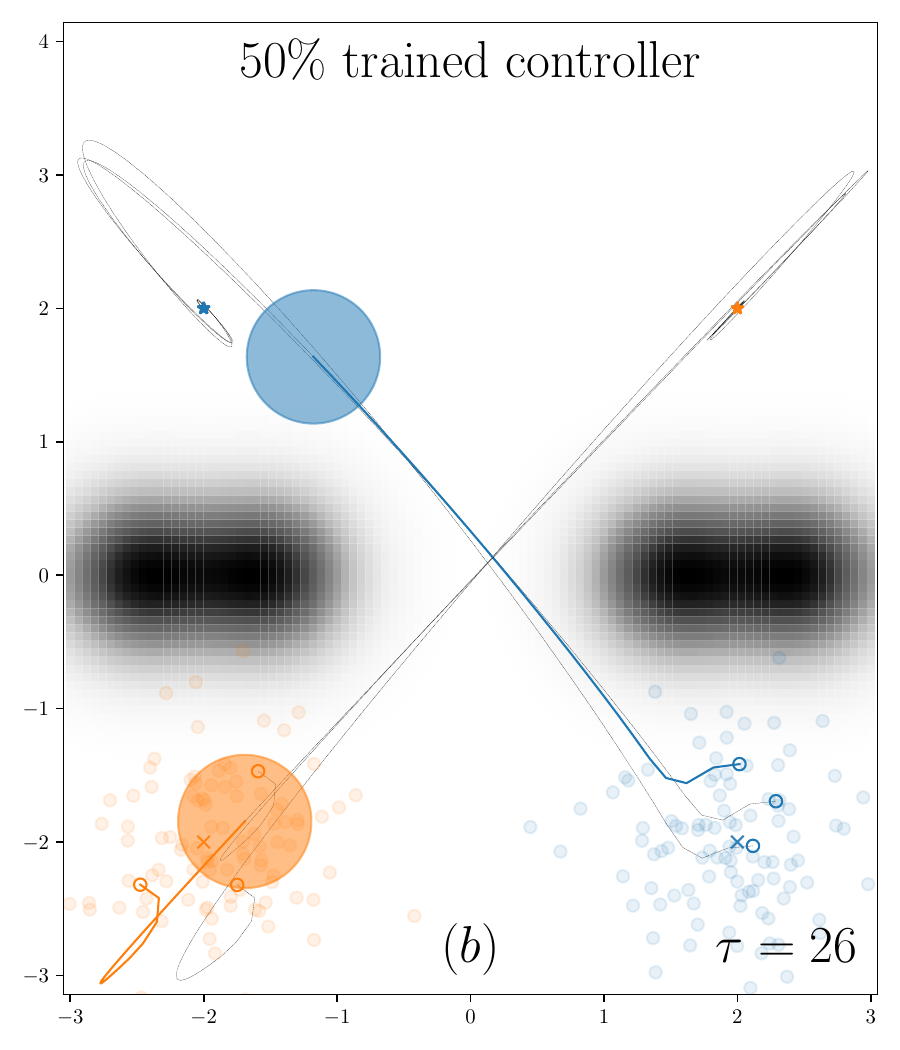}
	\end{minipage}%
	\begin{minipage}{0.32\linewidth}
	    \includegraphics[width=\linewidth]{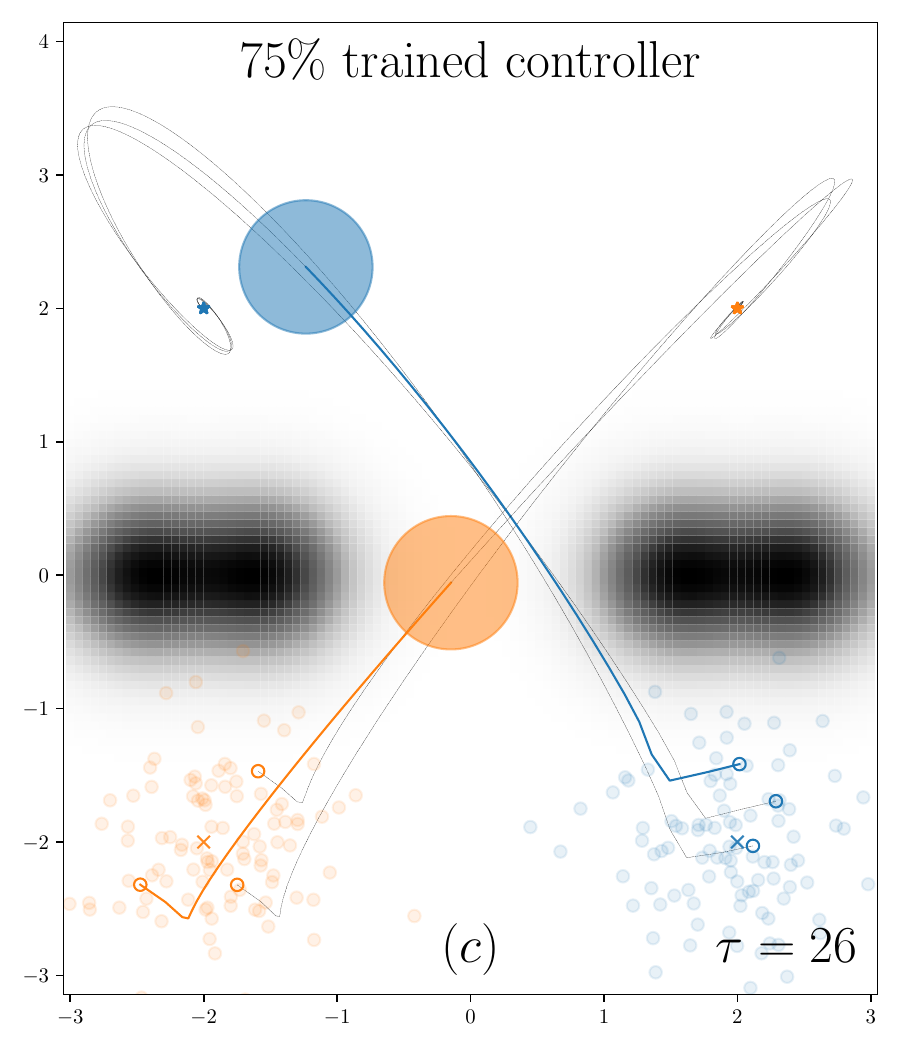}
	\end{minipage}
	\caption{\texttt{Mountains} --- Closed-loop trajectories after 25\%, 50\% and 75\% of the total training whose closed-loop trajectory is shown in Figure~\ref{fig:corridor}. Even if the performance can be further optimized, stability is always guaranteed.} 
	\label{fig:corridor_early_stopping}
\end{figure}

\subsubsection{Model mismatch}
We test our trained IMC controller when considering model mismatch on the system. In particular, we assume that the true vehicles have an incertitude over the mass of $\pm10\%$, and we apply IMC control policies embedding the nominal system with the nominal mass value. 
Figures~\ref{fig:corridor_modelMismatch_safe} (a-b) validate the robust $\ell_2$-stability of the closed-loop trajectories when the vehicles are lighter and heavier, respectively. Theorem~\ref{th:result_robust} suggests that, in this case, the gain of $\Emme$ may be sufficiently low to counteract the effect of model uncertainty. Note, however, that checking the sufficient condition \eqref{eq:condition_robustness} requires computing an upper bound on $\gamma(\bm{\Delta})$ --- a cumbersome task for general nonlinear systems.  Nonetheless, Theorem~\ref{th:result_robust} ensures that, in practical implementation, we can always reduce $\gamma(\Emme)$ enough to eventually meet \eqref{eq:condition_robustness}. 

\subsection{Boosting for safety and invariance certificates}
\label{sec:example_safe}
A challenging task in many control applications is to deal with stringent safety constraints on the state variables. 
Ideally, one would directly add the constraint that 
\begin{equation}
\label{eq:constraints_safety}
    x_t \in \mathcal{C}\,, \forall t=0,1,\ldots\,,
\end{equation}
in the IMC-based performance-boosting problem \eqref{eq:Prob_reform}, where 
$\mathcal{C}\subseteq \mathbb{R}^{n}$  defines a safety region. 
Unfortunately, \eqref{eq:constraints_safety} generally results in intractable constraints over $\Emme$. 
Indeed, it may be challenging to even verify that \eqref{eq:constraints_safety} holds for a certain $\Emme$ due to the infinite-horizon requirement and the involved nonlinearities.  
Many state-of-the-art approaches for guaranteeing safety hinge on either predictive safety filters \cite{hewing2020learning,wabersich2021predictive} or Control Barrier Functions (CBFs)~\cite{ames2019control, agrawal2017discrete}.
Safety filters are used during deployment: they override the control input 
$\mathbf{u} = \Emme(\hatbf{w})$ with a different (suboptimal) control variable when deemed necessary for guaranteeing safety. 
Instead, CBFs can be used for safety verification of a given policy, as they allow characterizing $\mathcal{C}$ as a forward invariant set based on a safety-set-defining function $h(x):\mathcal{X}\rightarrow\mathbb{R}$ satisfying $h(x)\geq 0$ for all $x\in\mathcal{C}$.  Certifying the forward invariance of $\mathcal{C}$ translates into determining if $h(x)$ is a CBF through verification of some safety conditions.%
\footnote{An exact definition of CBFs for the discrete-time can be found in%
~\cite{agrawal2017discrete}; for a more general discussion on CBFs we refer the reader to~\cite{ames2019control}.}
In particular, 
one can verify that, 
for any  $x_t\in\mathcal{C}$,
if there exists an input $u_t$ giving  $x_{t+1}$
such that it holds
\begin{equation}
    \label{eq:barrier_function}
    h(x_{t+1}) - h(x_t) + \gamma h(x_t) \geq 0 \,,
\end{equation}
where  $0<\gamma\leq 1$, then $h(x)$ is a CBF.
%
%

While optimizing over $\Emme$ such that \eqref{eq:constraints_safety} holds by design remains an open challenge, 
we aim to promote forward invariant sets
by shaping the cost to include soft safety specifications over a horizon of length $T$. In particular, the new cost term penalizes violations of~\eqref{eq:barrier_function} as per
\textit{
}
\begin{equation}
    \label{eq:loss_safe}
    L_{\operatorname{inv}} = \sum_{t=0}^{T-1} \operatorname{ReLU}\left(h(x_{t})-h(x_{t+1})+\gamma h(x_t)\right)\,.
\end{equation}

%
We consider the \texttt{mountains} scenario again and add the requirement that $(p_t^y)^{[i]}<(\bar{p}^y)^{[i]} + 0.1$
for each vehicle $i=1,2$ and every $t = 0,1,\ldots$, where $p_t^y$ denotes the $y$-coordinate of each center-of-mass position on the Cartesian plane. In other words, we only allow an overshoot of $0.1$ in the vertical direction with respect to the target position for each vehicle.
By defining $h(x_t) = \sum_{i=1}^2 ((\bar{p}^y)^{[i]} + 0.1 - (p_t^y)^{[i]})$ we add the term~\eqref{eq:loss_safe} to the loss function~\eqref{eq:learning_1}.
%
%
 Upon training without including $L_{\operatorname{inv}}$ in the cost, the masses violate the constraints, on average, on $67.49\%$ 
of the time over 100 runs --- typical trajectories are shown in Figure~\ref{fig:corridor}. The violation ratio is decreased to  $5.43\%$ when $L_{\operatorname{inv}}$ is included, as shown in Figure~\ref{fig:corridor_modelMismatch_safe}(c), where the gray area indicates the unsafe region to be avoided by the vehicles. Note that shaping the cost through $L_{\operatorname{inv}}$ is also beneficial if one implements an online safety filter such as \cite{hewing2020learning,wabersich2021predictive} during deployment. This is because penalizing $L_{\operatorname{inv}}$ drastically decreases constraint violations of the closed-loop system, and hence, the suboptimal online intervention of the safety filter would be much less frequent.


\begin{figure}
	\centering
	\begin{minipage}{0.32\linewidth}
	    \includegraphics[width=\linewidth]{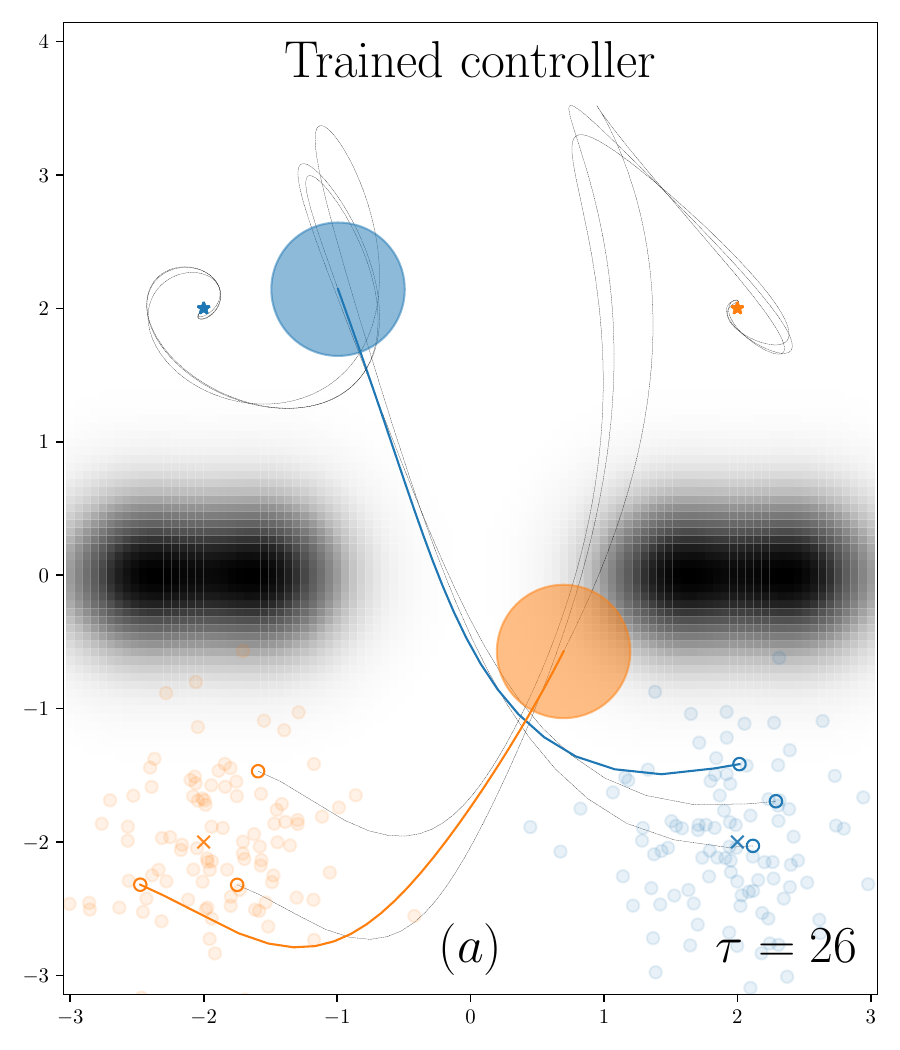}
	\end{minipage}%
	\begin{minipage}{0.32\linewidth}
	    \includegraphics[width=\linewidth]{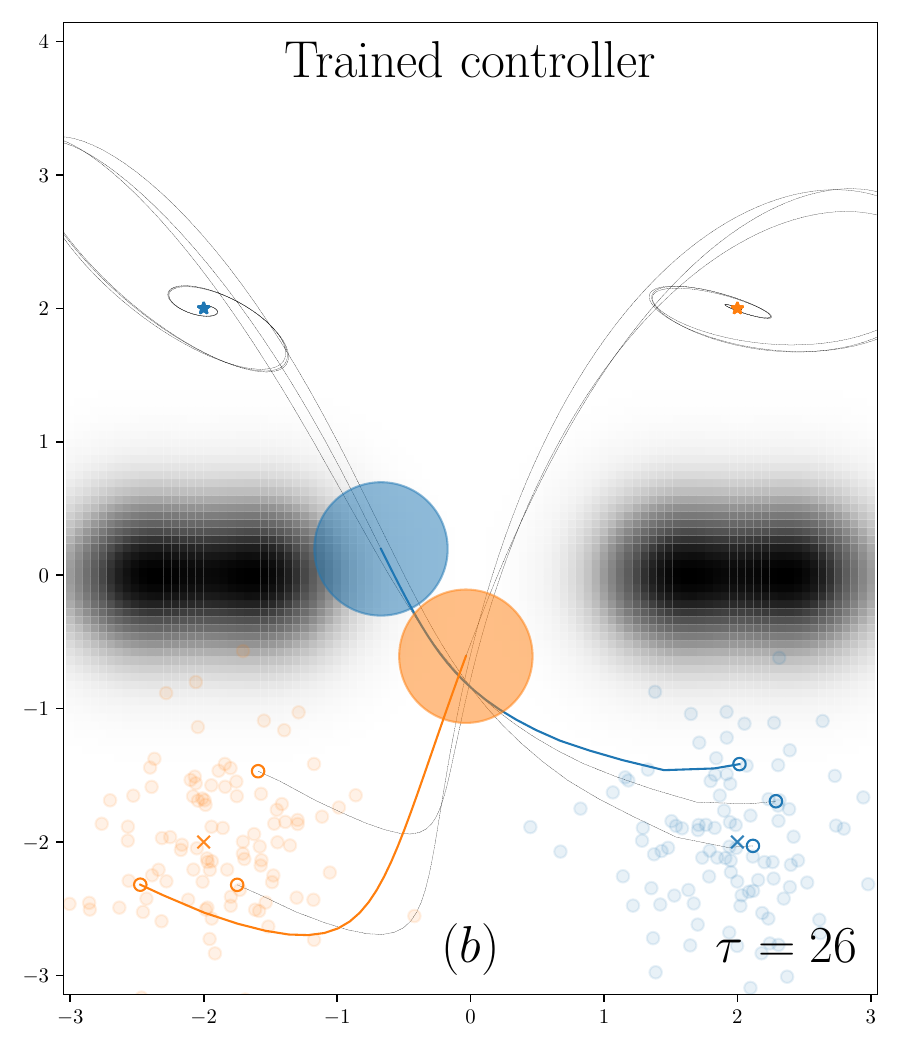}
	\end{minipage}%
    \begin{minipage}{0.32\linewidth}
	    \includegraphics[width=\linewidth]{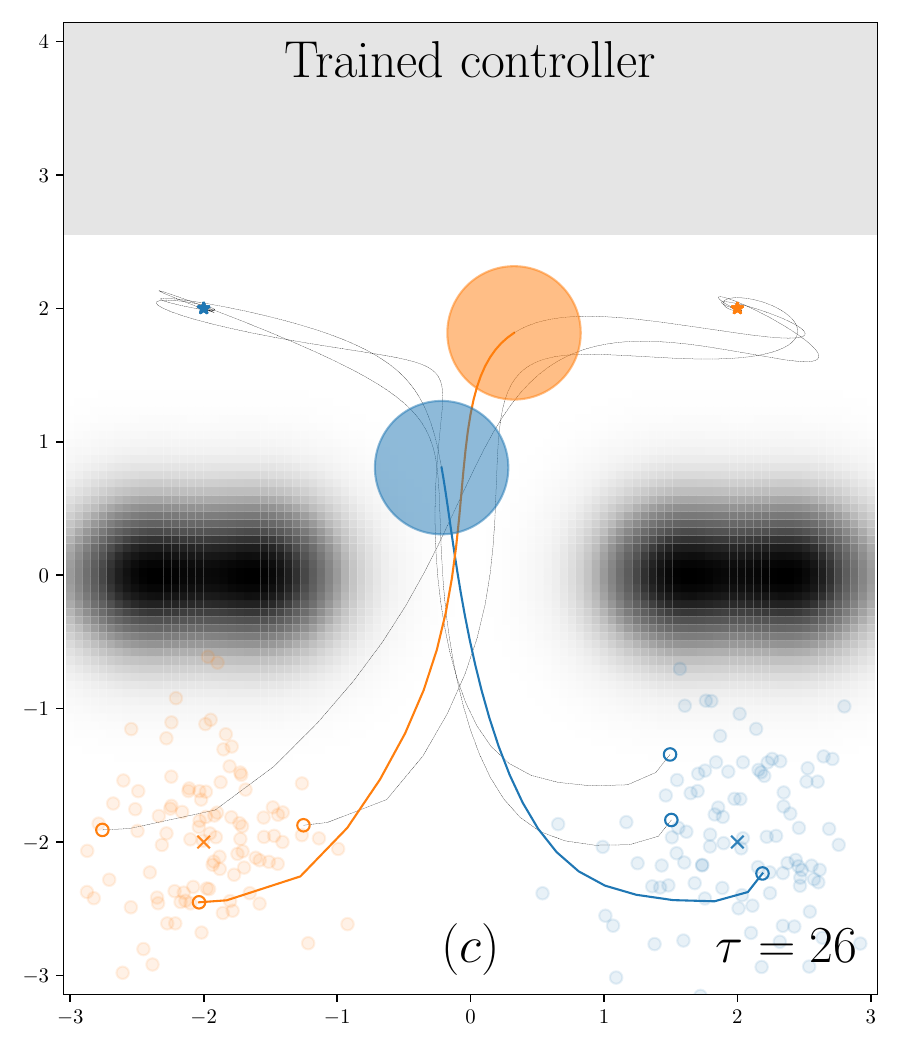}
	\end{minipage}
	\caption{\texttt{Mountains} --- Closed-loop trajectories after training. (Left and middle) Controller tested over a system with mass uncertainty (-10\% and +10\%, respectively).
    (Right) Trained controller with safety promotion through~\eqref{eq:loss_safe}. Training initial conditions marked with $\circ$. Snapshots taken at time-instants $\tau$. Colored (gray) lines show the trajectories in $[0,\tau_i]$ ($[\tau_i,\infty)$). Colored balls (and their radius) represent the agents (and their size for collision avoidance).} 
	\label{fig:corridor_modelMismatch_safe}
\end{figure}


\subsection{Boosting for temporal logic specifications}
The success of many policy learning algorithms, e.g., in RL, is highly dependent on the choice of the reward functions for capturing the desired behavior and constraints of an agent.
When tasks become complex, specifying loss functions that are the sum over time of stage costs can
be restrictive.
For instance, consider the case of an agent that must optimally 
visit a set of locations. 
A loss function composed of a stage-cost summed over time --- that is, the one considered in dynamic programming and classical optimal control \cite{mayne2000constrained,bertsekas2011dynamic}  --- cannot easily capture this task, as it would need a-priori information about the optimal timings to visit each location.
To overcome this problem, one could use more complex loss functions, as per those derived from temporal logic formulations.
In particular, truncated linear temporal logic (TLTL) 
is a specification language leveraging a set of operators defined over finite-time trajectories~\cite{li2017reinforcement,leung2023backpropagation}.
It allows incorporating domain knowledge, and constraints (in a soft fashion) into the learning process,
such as ``always avoid obstacles'', ``eventually visit location $a$'', or ``do not visit location $b$ until visiting location $a$''.
Then, using quantitative semantics one can automatically transform TLTL formulae into real-valued loss functions that are compositions of $\min$ and $\max$ functions over a finite period of time~\cite{li2017reinforcement, leung2023backpropagation}. 

To test the efficacy of TLTL specifications for shaping complex stable closed-loop behavior, we consider the scenario \texttt{waypoint-tracking}, shown in Figure~\ref{fig:waypoints}, where the two vehicles have to visit a sequence of waypoints while avoiding collisions between them and the gray obstacles. 
The \textit{blue} vehicle's goal is to visit
$g_{b}$, then $g_{a}$ and then $g_{c}$, 
while the goal for the \textit{orange} vehicle is to visit the waypoints in the following order:
$g_{c}$, $g_{b}$ and $g_{a}$.
Following~\cite{li2017reinforcement}, the loss formulation for the \textit{orange} agent is translated into plain English as 
``\textit{%
Visit $g_c$ then $g_b$ then $g_a$;
and don't visit $g_b$ or $g_a$ until visiting $g_c$;
and don't visit $g_a$ until visiting $g_b$;
and if visited $g_c$, don't visit $g_c$ again;
and if visited $g_b$, don't visit $g_b$ again;
and always avoid obstacles;
and always avoid collisions;
and eventually 
state at the final goal.%
}''
Its mathematical formulation can be found in Appendix~\ref{app:implementation}.\ref{subsec:ap:implementation_wp}.

Figure~\ref{fig:waypoints} shows the \texttt{waypoint-tracking} scenario before and after the training of a performance-boosting controller.
As described in Section~\ref{sec:implementation}.\ref{subsec:RENs}, we use a REN with $q=r=32$ for approximating the $\mathcal{L}_2$ operator $\Emme$. 
Furthermore, we allow for a time-varying bias of the form 
$b_t^\top  = \begin{bmatrix} 0_{1\times q} & 0_{1\times r} & b_{w,t}^\top \end{bmatrix}$, in~\eqref{eq:RENs}, with  $b_{w,t}=0$ for $t>T$.
While the system always starts at the same initial condition indicated with `$\circ$', the data consists of disturbance sequences $w_{T:0}$ with fixed $w_0$ and $w_{T:1}$ as i.i.d. samples drawn from a Gaussian distribution with zero mean and standard deviation of $0.01$.
Our result highlights the power of complex costs --- expressed through the TLTL loss function --- which promotes vehicles visiting the predefined waypoints in the correct order while avoiding collisions between them and with the obstacles.

\begin{figure}
	\centering
	\begin{minipage}{0.32\linewidth}
	    \includegraphics[width=\linewidth]{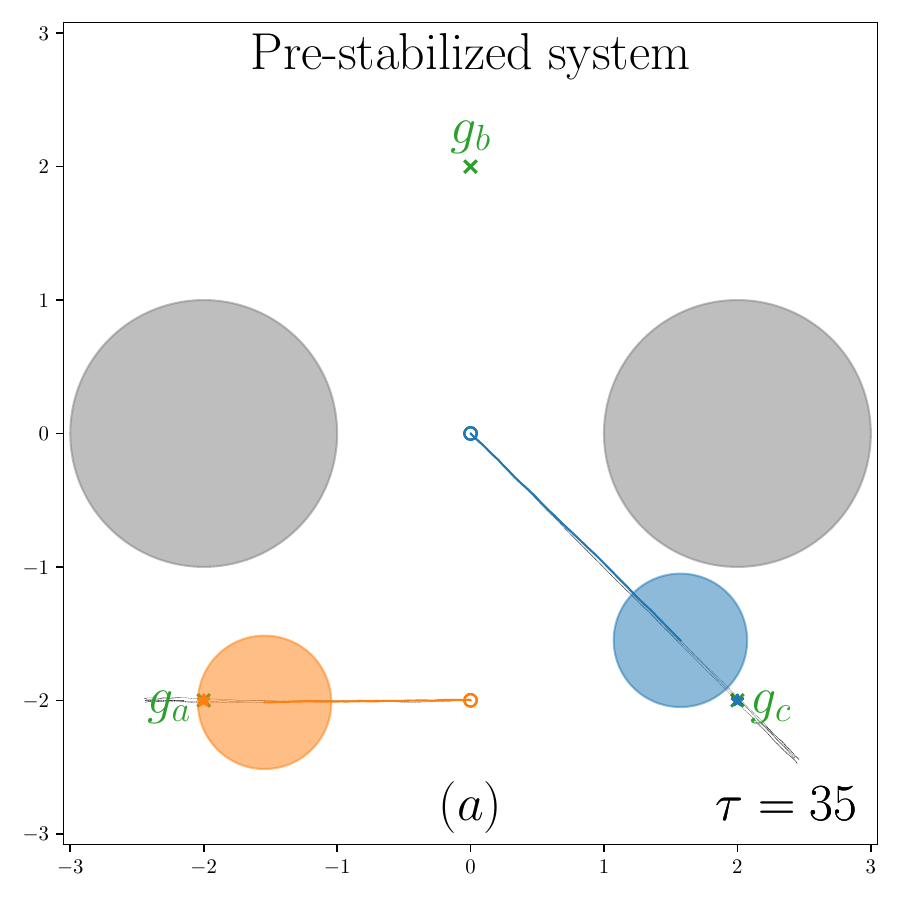}
	\end{minipage}%
	\begin{minipage}{0.32\linewidth}
	    \includegraphics[width=\linewidth]{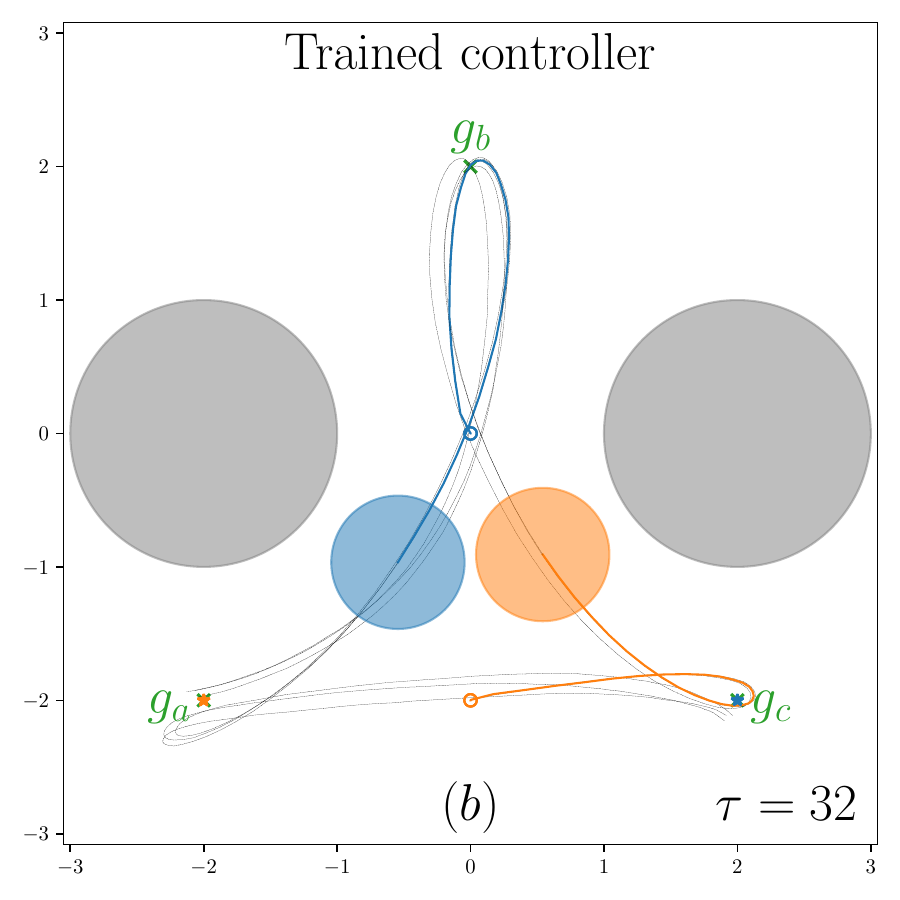}
	\end{minipage}%
	\begin{minipage}{0.32\linewidth}
	    \includegraphics[width=\linewidth]{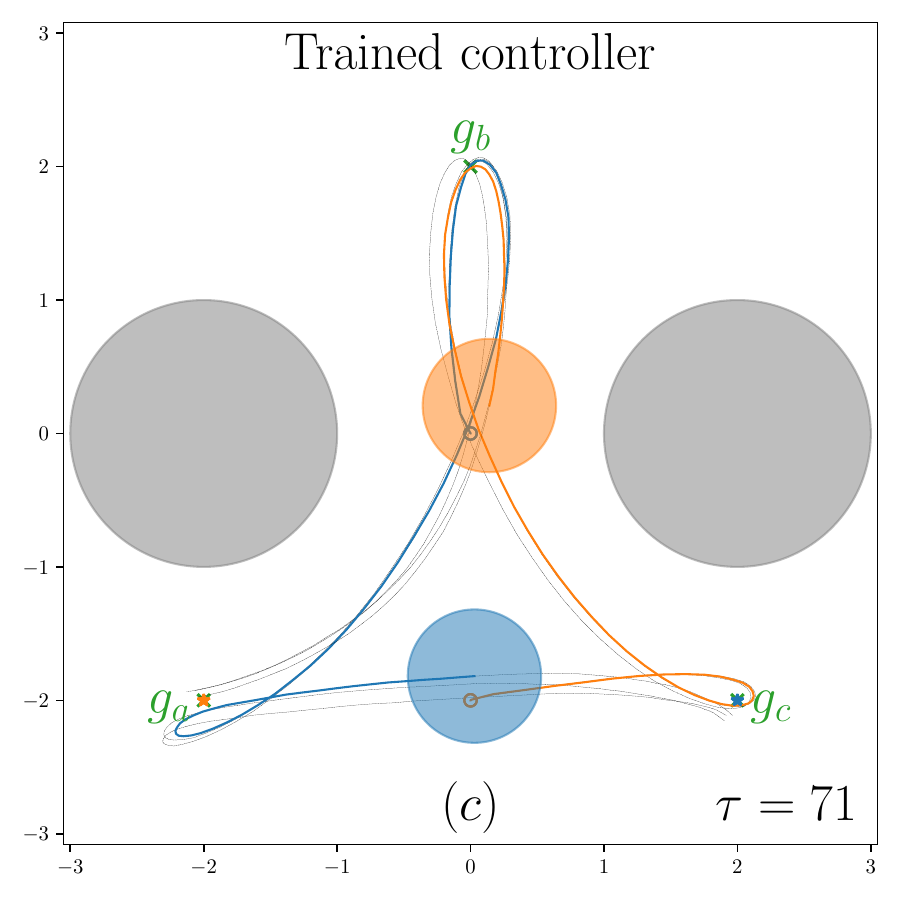}
	\end{minipage}
	\caption{\texttt{Waypoint-tracking} --- Closed-loop trajectories before training (left) and after training (middle and right). Snapshots taken at time-instants $\tau$. Colored (gray) lines show the trajectories in $[0,\tau_i]$ ($[\tau_i,\infty)$). Colored balls (and their radius) represent the agents (and their size for collision avoidance).} 
	\label{fig:waypoints}
\end{figure}

\section{Conclusion}

Embedding safety and stability emerges as a crucial challenge when control systems are equipped with high-performance machine learning components. This work aims to contribute to this rapidly developing field by uncovering the theoretical and computational potential of IMC for safely boosting the performance of closed-loop nonlinear systems with machine learning models such as DNNs. 

The results of this work open up several future research directions. \textcolor{black}{First, motivated by the recent results of \cite{boroujeni2024pac}, it would be relevant to apply statistical learning theory to rigorously assess the generalization capabilities of performance-boosting controllers in uncertain environments, with uncertain models, and over extended time frames.} Second, drawing on insights from \cite{bertsekas2022lessons}, integrating extensive RL-based offline learning with real-time adjustments similar to MPC presents a promising approach. Third, within the IMC framework, \textcolor{black}{there is a significant opportunity to develop richer parametrizations of stable dynamical systems in $\mathcal{L}_p$, and to theoretically prove their approximation capabilities}. Lastly, building upon \cite{martin2024learning}, it is interesting to explore how learning-based IMC methods could generate new optimization algorithms with formal guarantees for tackling complex optimal control and machine learning tasks.

\appendix
\subsection{Implementation details for the numerical experiments in Section \ref{sec:numerical}}
\label{app:implementation}

We set $m^{[i]}=b^{[i]}_1={k'}^{[i]}_1={k'}^{[i]}_2=1$ and $b^i_2=0.5$ as the parameters for each vehicle $i$, in the model~\eqref{eq:mechanical_system} with the pre-stabilizing controller~\eqref{eq:pre_controller}. The collision-avoidance radius of each agent is 0.5.

\subsubsection{Mountains scenario}
As shown in Figure~\ref{fig:corridor}, the vehicles start at $p^{[1]}_{0} = (-2,-2)$ and $p^{[2]}_{0} = (-2,2)$, and their goal is to go to the target positions $\bar{p}^{[1]} = (2,2)$ and $\bar{p}^{[2]} = (-2,2)$, respectively.
The training data consists of $100$ initial positions sampled from a Gaussian distribution around the initial position with a standard deviation of $0.5$.

Let $\bar{x} = (\bar{x}^{[1]}, \bar{x}^{[2]})$ with $\bar{x}^{[i]} = (\bar{p}^{[i]}, 0_2)$.
The terms of the cost function~\eqref{eq:loss_CA} are defined as follows:
\begin{align*}
    l_{traj}(x_t,u_t) = (x_t-\bar{x})^\top \tilde{Q} (x_t-\bar{x}) + \alpha_u u_t^\top u_t
\end{align*}
\begin{equation*}
l_{ca}(x_t) =
\begin{cases}
\alpha_{ca}
\sum_{i=0}^{N}
\sum_{j,\,i\ne j} (d^{i,j}_{t} + \epsilon)^{-2} 
&
\text{if} \, d^{i,j}_{t} \leq D_{\text{safe}}\,, \\ 
0 
&
\text{otherwise}\,,
\end{cases}
\end{equation*}
where 
$\tilde{Q}\succ 0$ and $\alpha_u,\alpha_{ca}>0$ are hyperparameters,
$d^{i,j}_{t} = |p^{[i]}_t-p^{[j]}_t|_2 \geq 0$ denotes the distance between agent $i$ and $j$, $\epsilon>0$ is a fixed positive small constant such that the loss remains bounded for all distance values and $D_{\text{safe}}$ is a safe distance between the center of mass of each the agent; we set it to 1.2.

\allowdisplaybreaks

Motivated by \cite{onken2021neural}, we represent the obstacles based on a Gaussian density function 
\begin{equation*}
\eta(z; \mu, \Sigma) = \frac{1}{2\pi \sqrt{\text{det}(\Sigma)}} \exp\left( -\frac{1}{2} \left(z-\mu\right)^\top \Sigma^{-1} \left(z-\mu\right) \right)\,,
\end{equation*}
with mean $\mu \in \mathbb{R}^2$ and covariance $\Sigma \in \mathbb{R}^{2\times 2}$ with $\Sigma\succ 0$.
The term $l_{obs}(x_t)$ is given by
\begin{align}
l_{obs}(x_t) 
=
\alpha_{obs} 
\sum_{i=0}^{2}
&\Bigg(
\eta\left(p^{[i]}_t; \begin{bmatrix}2.5\\0\end{bmatrix}, 0.2\,I\right) \nonumber\\
&~~+
\eta\left(p^{[i]}_t; \begin{bmatrix}-2.5\\0\end{bmatrix}, 0.2\,I\right)  \nonumber\\
&~~+
\eta\left(p^{[i]}_t; \begin{bmatrix}1.5\\0\end{bmatrix}, 0.2\,I\right) \nonumber\\
&~~+
\eta\left(p^{[i]}_t; \begin{bmatrix}-1.5\\0\end{bmatrix}, 0.2\,I\right)
\Bigg) \,.
\label{eq:l_obs}
\end{align}

For the hyperparameters, we set 
$\alpha_u = 2.5\times 10^{-4}$, 
$\alpha_{ca} = 100$, 
$\alpha_{obs} = 5 \times 10^{3}$ and 
$Q=I_4$.
We use stochastic gradient descent with Adam to minimize the loss function, setting a learning rate of $1\times 10^{-4}$.
We train for $5\times 10^{3}$ epochs with one trajectory per batch size.

\subsubsection{Waypoint-tracking scenario}
\label{subsec:ap:implementation_wp}

As shown in Figure~\ref{fig:corridor}, the vehicles start at $p^{[1]}_{0} = (-2,0)$ and $p^{[2]}_{0} = (0,0)$.
The goal points $g_a$, $g_b$ and $g_c$ are located at $(-2,-2)$, $(0,2)$ and $(2,-2)$, respectively.
To describe the TLTL loss, let us define, for each vehicle, the following functions of time:
\begin{itemize}
    \item $d^{g_i}_t$, for $i=1,2,3$, is the distance between the vehicle and the goal point $g_i$;
    \item $d^{o_i}_t$, for $i=1,2$, is the distance between the vehicle and the $i^{\text{th}}$ obstacle;
    \item $d^{coll}_t$ is the distance between the two vehicles;
\end{itemize}
where $g_1$, $g_2$ and $g_3$ are the waypoints in the correct visiting order, for each vehicle.
Following the notation of~\cite{li2017reinforcement}, 
the  temporal logic form of the cost function, for each vehicle, is
\begin{multline}
    \label{eq:TLTL_cost}
    \left(\psi_{g_1} \,\mathcal{T}\,\psi_{g_2} \,\mathcal{T}\,\psi_{g_3} \right)
    \wedge 
    \left(\lnot\left(\psi_{g_2} \vee \psi_{g_3}\right) \,\mathcal{U}\, \psi_{g_1}\right) 
    \wedge 
    \left(\lnot\psi_{g_3}\,\mathcal{U}\,\psi_{g_2}\right) 
    \\
    \wedge 
    \left(\bigwedge_{i=1,2,3}\square\left(\psi_{g_i}\Rightarrow\bigcirc\square\lnot\psi_{g_i}\right)\right) 
    \wedge 
    \left(\bigwedge_{i=1,2}\square\psi_{o_i}\right) 
    \\
    \wedge
    \square\psi_{coll}
    \wedge
    \lozenge\square\psi_{g_3} 
\end{multline}
where 
$\psi$ are predicates defined in Table~\ref{tab:predicatesTLTL}, and $r_{obs} = 1.7$ and $r_{r} = 0.5$ are the radii of the obstacles and vehicles, respectively.%
\footnote{Note that in the waypoint-tracking scenario, we do not model the obstacles with a Gaussian density function.}
The Boolean operators $\lnot$, $\vee$, and $\wedge$ stand for negation (not), disjunction (or), and conjunction (and). The temporal operators $\mathcal{T}$, $\mathcal{U}$, $\lozenge$, and $\square$  stand for `then', `until', `eventually', and `always'. Mathematically, each term can be automatically translated following~\cite{li2017reinforcement,leung2023backpropagation}.
For instance,
$\square\psi_{coll}$ translates into 
\begin{equation*}
    \min_{t\in[0,T]}  (d^{rob}_t-2r_{rob}) ,
\end{equation*}
and $\square\left(\psi_{g_i}\Rightarrow\bigcirc\square\lnot\psi_{g_i}\right)$ translates into
\begin{align*}
    \min_{t\in[0,T]} \max \big( 
    \begin{aligned}[t]
        & -(0.05-d^{g_i}_t) \,,\, 
        &\hspace{-2pt}\min_{\tilde{t}\in[t+1,T]} -(0.05-d^{g_i}_t) \big) .
    \end{aligned}
\end{align*}
The full mathematical expression of~\eqref{eq:TLTL_cost}, which can be obtained following~\cite{li2017reinforcement}, is implemented
in our \href{https://github.com/DecodEPFL/performance-boosting_controllers.git}{Github repository}.

\begin{table}[!bt]
    \centering
    \begin{tabular}{c|c}
        Predicates & Expression \\
        \hline
        $\psi_{g_1}$ & $d^{g_1} < 0.05$ \\
        $\psi_{g_2}$ & $d^{g_2} < 0.05$ \\
        $\psi_{g_3}$ & $d^{g_3} < 0.05$ \\
        $\psi_{o_1}$ & $d^{o_1} > r_{obs}$ \\
        $\psi_{o_2}$ & $d^{o_2} > r_{obs}$ \\
        $\psi_{coll}$ & $d^{rob} > 2\,r_{rob}$
    \end{tabular}
    \caption{Predicates used in the TLTL formulation of \eqref{eq:TLTL_cost}.}
    \label{tab:predicatesTLTL}
\end{table}

We also add a small regularization term for promoting that the vehicles stay close to the end target point, which reads $\alpha_{\text{reg}}\norm{x_t-\bar{x}}^2$, with $\alpha_{\text{reg}}=1\times10^{-4}$.
We use stochastic gradient descent with Adam to minimize the loss function, setting a learning rate of $5\times10^{-4}$.
We train for 3000 epochs with a single trajectory per batch size.

\subsection{Comparison of performance-boosting controllers with other baselines}
\label{app:comparison}


We compare the performance of our proposed controllers with two baseline approaches for the scenario \texttt{mountains} presented in Section~\ref{sec:numerical}.\ref{sec:example_stab}.
In both cases, the vehicles are equipped with the base proportional controller \eqref{eq:pre_controller} which is able to steer the agents towards the target position in a stable way.
As described in Section~\ref{sec:numerical}.\ref{sec:example_stab}, improving the performance means vehicles must avoid collisions with each other and with obstacles.

 The first baseline we consider is a control policy  derived by solving an optimization problem in a receding-horizon manner. This optimization problem is defined over the set of control inputs that ensure collision avoidance within the horizon. 

The second baseline is to directly parametrize the entire control policy $\mathbf{u} = \mathbf{K}(x)$ as a recurrent neural network, that is, without adopting the IMC architecture of Figure~\ref{fig:IMCscheme} train a control policy $\mathbf{u}= \mathbf{K}(\mathbf{x})$ directly parametrized as a recurrent neural network optimizing the cost $L(x_{T:0},u_{T:0})$ defined in Section~\ref{sec:numerical}.\ref{sec:example_stab}.
Note that this approach does not guarantee the stability of the resulting closed-loop system.

\subsubsection{Online-optimization using barrier functions over the base controller}

A common approach in robotics for avoiding collisions and unsafe regions is to use control barrier functions~\cite{ames2019control,agrawal2017discrete}.


This requires online optimization for computing the system inputs. Specifically, we consider the approach in~\cite{agrawal2017discrete} for guaranteeing that the safe region is forward invariant.
The online optimization problem reads as
\begin{subequations}
\label{eq:online_opt}
\begin{align}
u_t^* = \arg~&\min_{u_t,u_{t+1}}   u_t^\top u_t \\
\text{s.t.} ~~& x_{t} = (p_t^{[1]},q_t^{[1]},p_t^{[2]},q_t^{[2]})\,,\\
&u_t = (u_t^{[1]},u_t^{[2]})\,,~\eqref{eq:mechanical_system},\eqref{eq:pre_controller},\eqref{eq:control_input_additional}\,, \\
&h(x_{t+1}) - h(x_t) + \gamma\, h(x_t) \geq 0\,,\\
&h(x_{t+2}) - h(x_{t+1}) + \gamma\, h(x_{t+1}) \geq 0\,,
\label{eq:online_opt_constraint}
\end{align}
\end{subequations}
where $0 <\gamma\leq 1$ and 
$u_t^*$ is the safety-preserving input to the system.
The barrier function $h: \mathbb{R}^n\rightarrow \mathbb{R}$ characterizes the region $\mathcal{C} = \{x \in \mathbb{R}^n : h(x)\geq0\}$ in the state space 
where no collisions between agents nor with the obstacles occur. 
To this purpose, we define
\begin{multline*}
    h(x) = \left(|p^{[1]}-p^{[2]}|^2 - 4 \, r_{\text{agent}}^2 \right)\\
    + \sum_{i=1}^2 \sum_{j=1}^4 \left(|p^{[i]}-p^{\text{obs}_j}|^2 - r_{\text{obs}}^2 \right)\,,
\end{multline*}
which is positive in the safe region. 
The radius of each agent is $r_{\text{agent}}=0.5$, while $r_{\text{obs}}=1.4$ denotes the radius of two obstacles, modeled as disks. The center of each disk is given by $p^{\text{obs}_j}\in\mathbb{R}^2$ and is the mean of the Gaussian density functions used for defining $l_{obs}$ in \eqref{eq:l_obs}. 
In the absence of collisions at time $t$, the constraint \eqref{eq:online_opt_constraint} forces the agents to stay in the safe region at time $t+1$ as well.

Figure~\ref{fig:online_optimization} shows three closed-loop trajectories of the agents when starting from different initial conditions. 
When the initial positions are symmetric with respect to the $y$-axis (Figure~\ref{fig:online_optimization}(a)), the optimization problem~\eqref{eq:online_opt} cannot find an input $u_t^*$ 
allowing both agents to pass through the narrow corridor, and the agents stop without reaching the target.
This is due to the reactive nature of CBFs, which do not account for the behavior of the system nor prioritize target reaching.
To provide an objective performance assessment,  we compare the quadratic cost term on the state, i.e. we evaluate the Euclidean distance to the target using
\begin{equation}
    \label{eq:loss_comparison}
    \tilde{L}(x_{T:0},u_{T:0}) = \sum_{t=0}^T (x_t-\bar{x})^\top (x_t-\bar{x})\,,
\end{equation}
over 20 test initial conditions
(sample trajectories are displayed in Figure~\ref{fig:online_optimization}(b-c)).
The average cost incurred by the control law is $25.81$, while it is $20.94$ when using our approach. We highlight that when the vehicles are close enough to their respective target positions, one has $u_t^\star = 0$, and the system inherits the stability properties due to the base controller.

\begin{figure}
	\centering
	\begin{minipage}{0.32\linewidth}
	    \includegraphics[width=\linewidth]{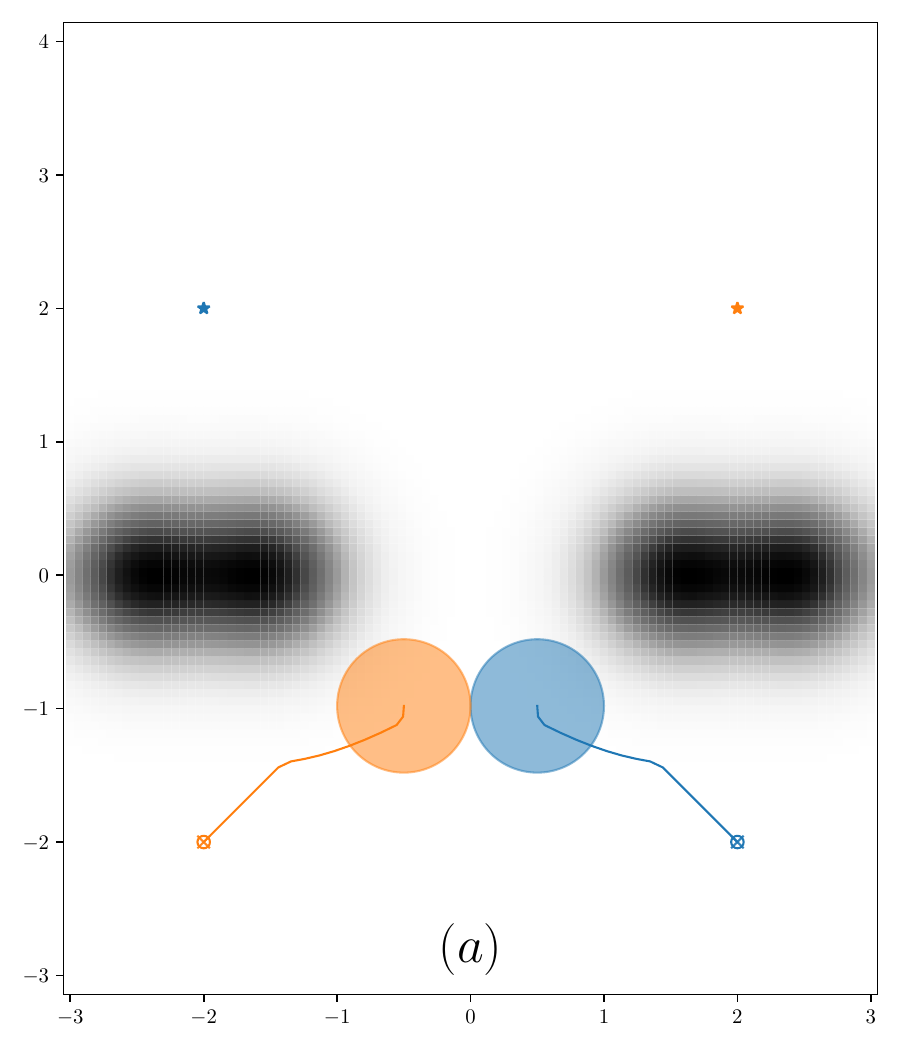}
	\end{minipage}%
	\begin{minipage}{0.32\linewidth}
	    \includegraphics[width=\linewidth]{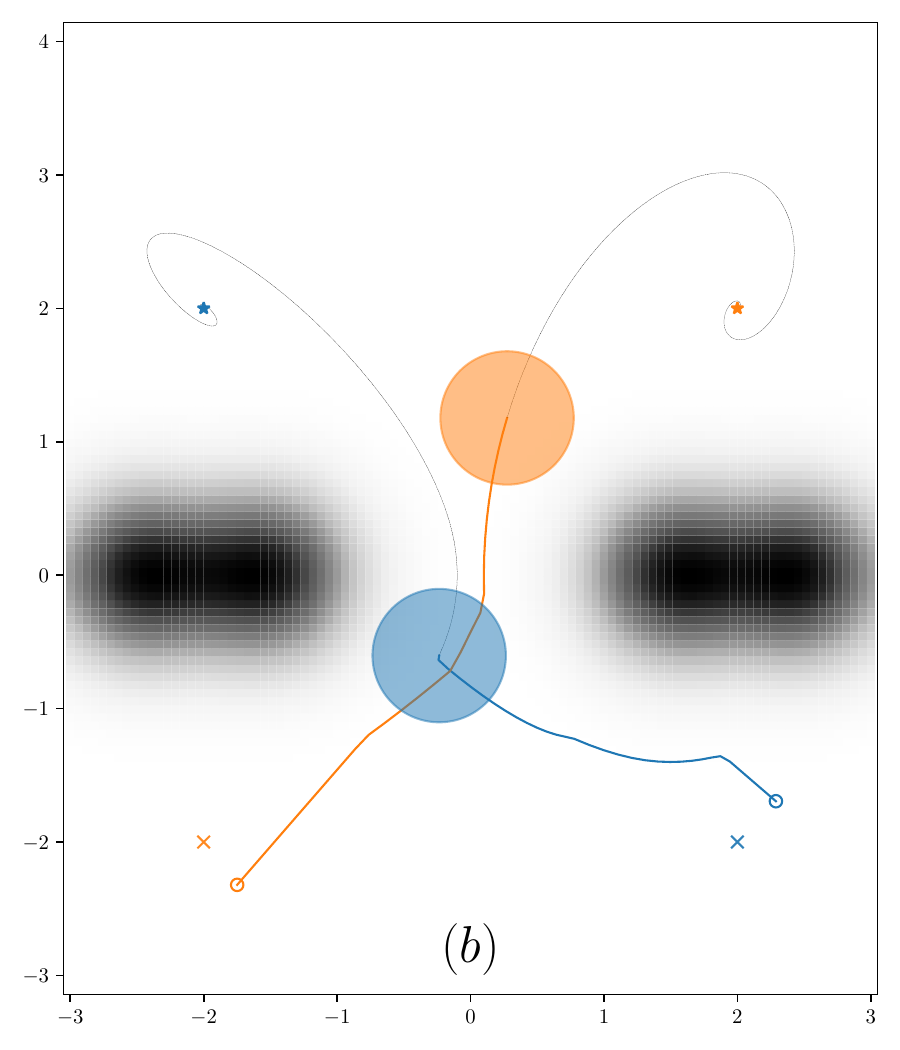}
	\end{minipage}%
	\begin{minipage}{0.32\linewidth}
	    \includegraphics[width=\linewidth]{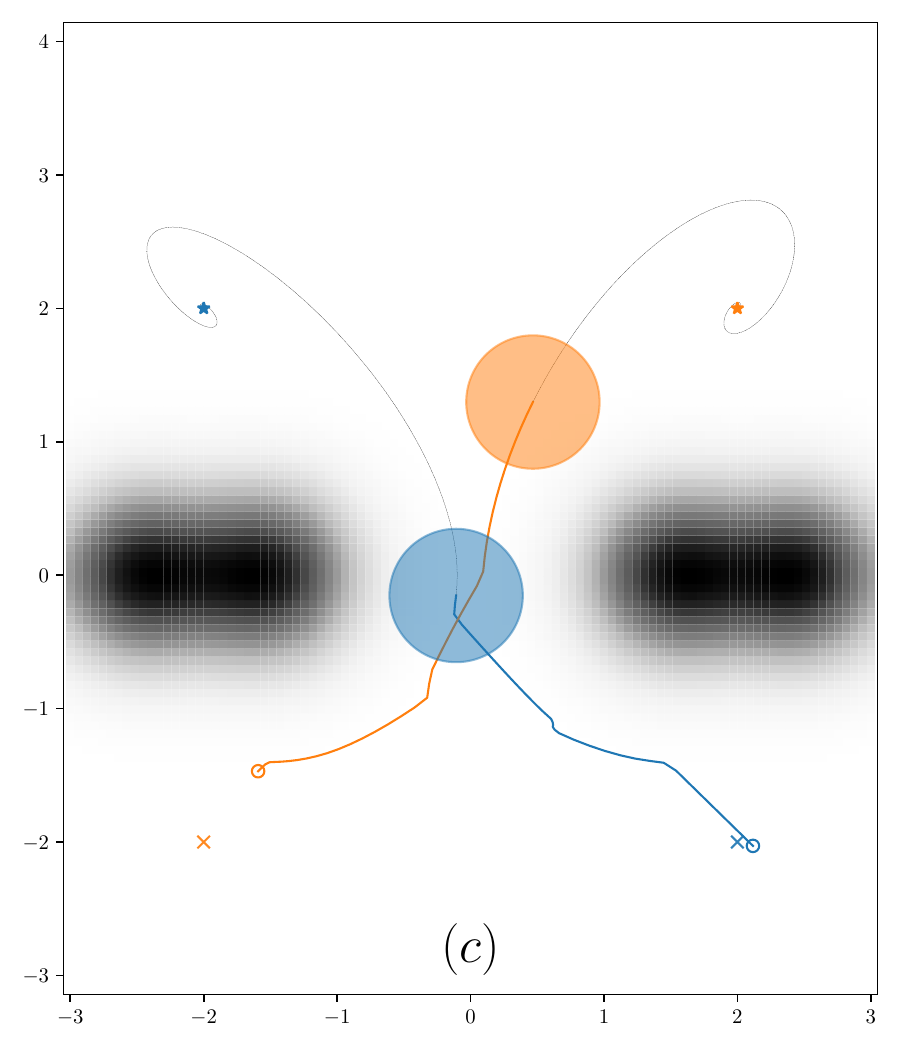}
	\end{minipage}
	\caption{\texttt{Mountains} --- Closed-loop trajectories when using the online policy given by \eqref{eq:online_opt}.
     Snapshots of three trajectories starting at different test initial conditions.} 
	\label{fig:online_optimization}
\end{figure}

\subsubsection{A recurrent neural network controller}

We replace the controller in Figure~\ref{fig:IMCscheme} by a REN where the trainable parameters are the weights $W$ and the time-invariant bias $b_t = b$ in \eqref{eq:RENs}.
Note that we do not constrain the REN to be an $\mathcal{L}_2$  operator, i.e., we do not use the mapping $\Theta$ described in Section~\ref{sec:implementation}.\ref{subsec:RENs} for redefining the trainable parameters.
The model consists of 861 parameters which are optimized for minimizing the cost $L(x_{T:0},u_{T:0})$, using the same initial conditions as in the experiments of Section~\ref{sec:numerical}.\ref{sec:example_stab}.
Figure~\ref{fig:RNN} shows three closed-loop trajectories of the agents when starting from different initial positions. Note that the targets are no longer the equilibria of the closed-loop system, and the vehicles move away from the targets after an initial reaching phase.
The cost \eqref{eq:loss_comparison} incurred by this control law is $26.60$, while it is $20.94$
when using 
a performance-boosting controller
(where the REN representing the operator $\Emme$ has 864 parameters, i.e., only three more than the above REN controller).

\begin{figure}
	\centering
	\begin{minipage}{0.32\linewidth}
	    \includegraphics[width=\linewidth]{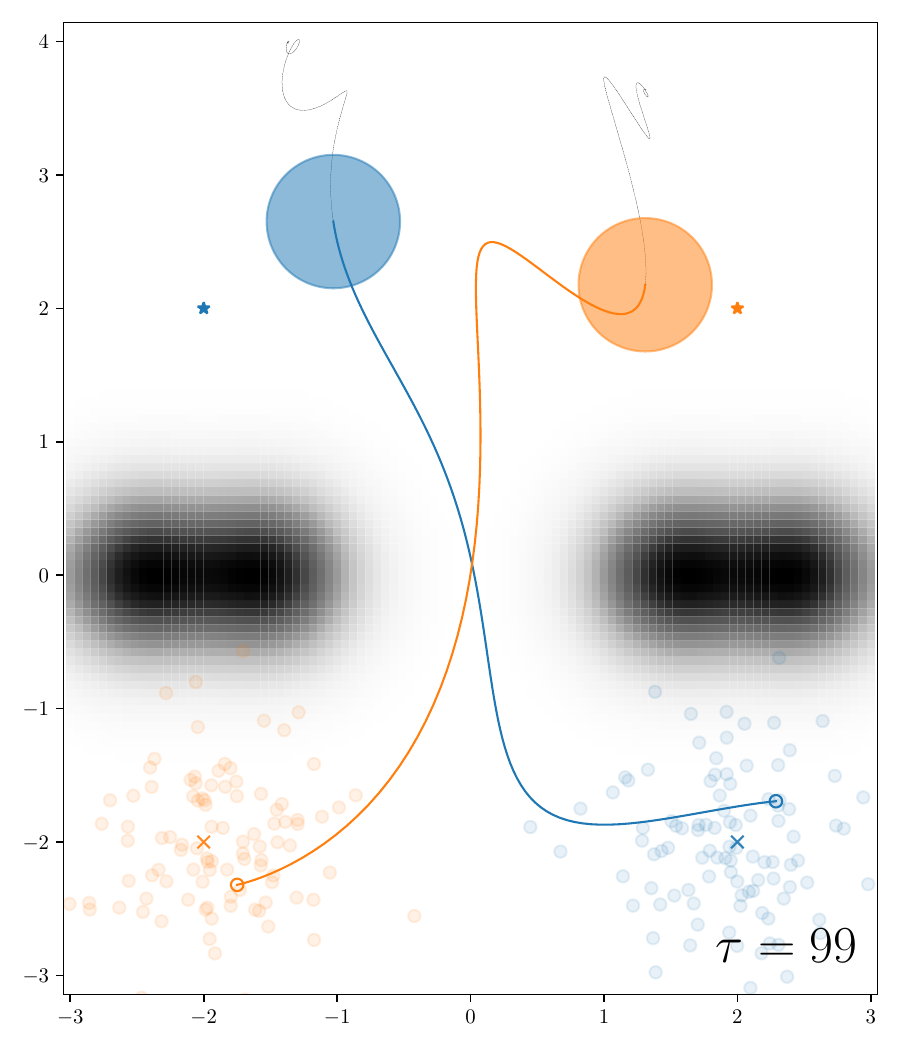}
	\end{minipage}%
	\begin{minipage}{0.32\linewidth}
	    \includegraphics[width=\linewidth]{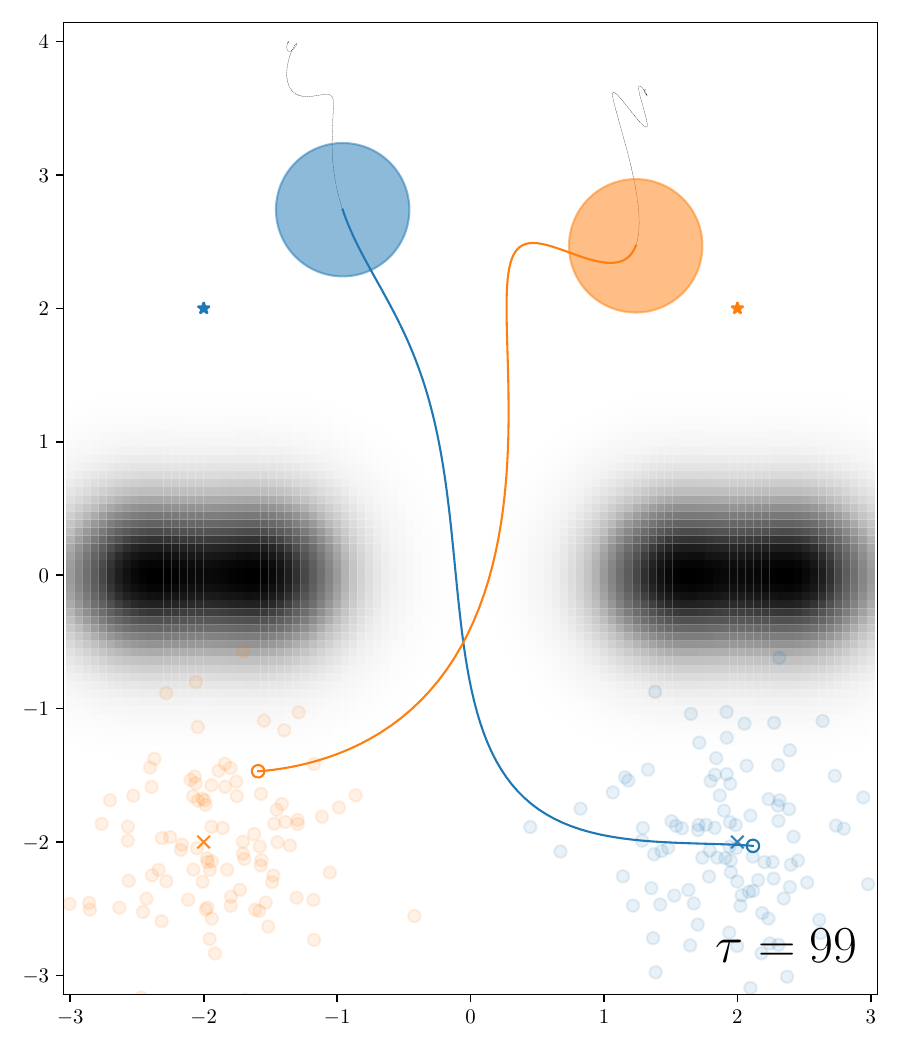}
	\end{minipage}%
	\begin{minipage}{0.32\linewidth}
	    \includegraphics[width=\linewidth]{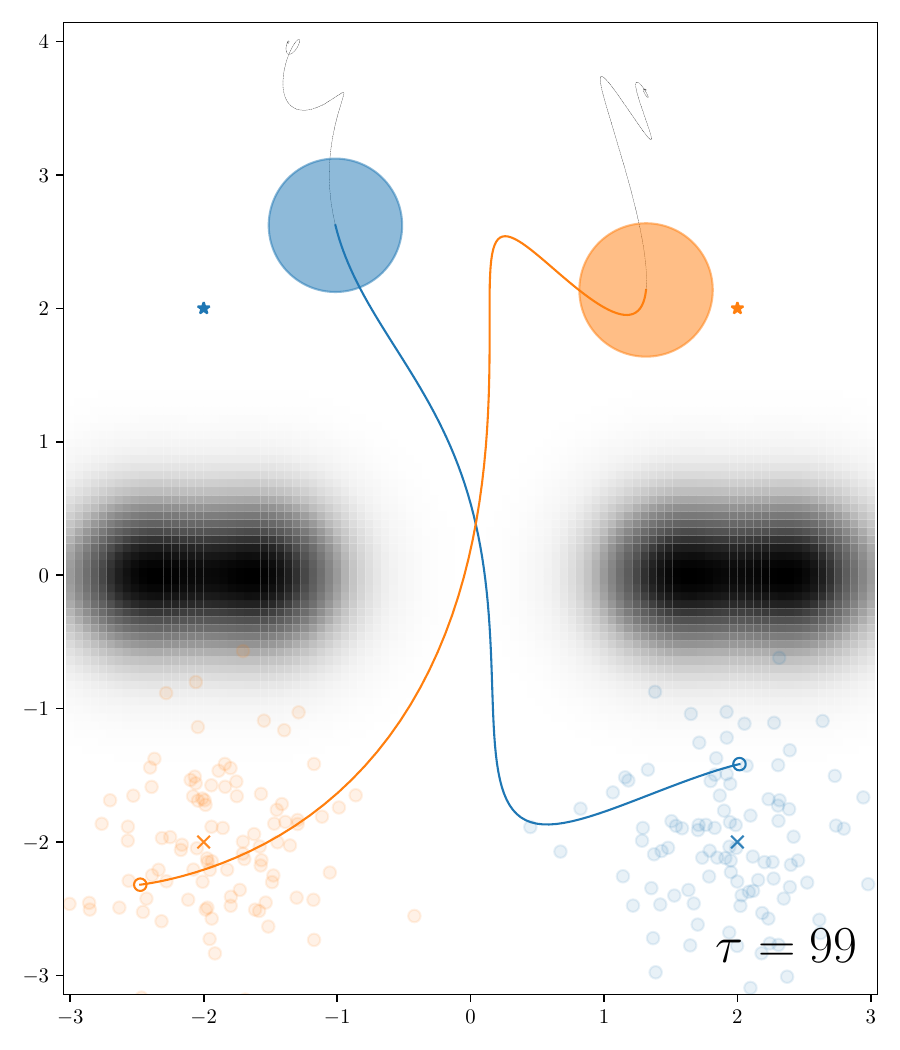}
	\end{minipage}
	\caption{\texttt{Mountains} --- Three different closed-loop trajectories after training a REN controller without $\mathcal{L}_2$ stability guarantees over 100 randomly sampled initial conditions marked with $\circ$. Colored (gray) lines show the trajectories in (after) the training time interval.} 
	\label{fig:RNN}
\end{figure}





\bibliographystyle{IEEEtran}
\bibliography{references.bib}

\end{document}